\documentclass[english]{article}
\usepackage[T1]{fontenc}
\usepackage{babel}
\usepackage{array}
\usepackage{url}
\usepackage{multirow}
\usepackage{amsmath}
\usepackage{amsthm}
\usepackage{graphicx}
\usepackage{epstopdf}
\usepackage{epsfig}
\usepackage{tabularx}
\usepackage[authoryear]{natbib}
\usepackage[unicode=true,pdfusetitle,
bookmarks=true,bookmarksnumbered=false,bookmarksopen=false,
breaklinks=false,pdfborder={0 0 1},backref=false,colorlinks=false]
{hyperref}
\hypersetup{
	colorlinks,urlcolor=blue,linkcolor=blue,anchorcolor=blue,citecolor=blue}
\usepackage{breakurl}

\makeatletter

\theoremstyle{plain}
\newtheorem{thm}{\protect\theoremname}
\theoremstyle{remark}
\newtheorem{rem}{\protect\remarkname}
\theoremstyle{plain}

\theoremstyle{definition}

\theoremstyle{plain}
\newtheorem{lem}{\protect\lemmaname}

\newcommand{\be}{\begin{equation}}
\newcommand{\ee}{\end{equation}}
\newcommand{\bey}{\begin{eqnarray}}
\newcommand{\eey}{\end{eqnarray}}
\newcommand{\beyn}{\begin{eqnarray*}}
\newcommand{\eeyn}{\end{eqnarray*}}
\usepackage{color}
\usepackage{soul,xcolor}
\setstcolor{red}
\newif\ifApproveEdit
\ApproveEditfalse
\ifApproveEdit
  
  \newcommand{\DEL}[1]{\iffalse{#1}\fi}
\else
  
  \newcommand{\DEL}[1]{\st{#1}}
\fi

\usepackage{graphics}
\usepackage{epstopdf}

\renewcommand{\baselinestretch}{1.5} 

\usepackage{bm}

\usepackage[noblocks]{authblk}
\def\singlespace{\def\baselinestretch{1}\@normalsize}

\makeatother

\newcommand{\hM}{\hat{M}}

\newcommand{\bK}{\bm{K}}
\newcommand{\tbK}{\tilde{\bK}}

\newcommand{\bL}{\bm{L}}
\newcommand{\tbL}{\tilde{\bL}}
\newcommand{\calR}{\mathcal{R}}  
\newcommand{\calY}{\mathcal{Y}}

\usepackage{soul,xcolor,cancel}
\usepackage{color}

\setstcolor{red}
\newif\ifApproveEdit
\ApproveEditfalse
\ifApproveEdit

\newcommand{\del}[1]{\iffalse{#1}\fi}
\else

\newcommand\del[2][red]{\setbox0=\hbox{$#2$}%
\rlap{\raisebox{.45\ht0}{\textcolor{#1}{\rule{\wd0}{1pt}}}}#2}
\fi

\newcommand{\limsum}{\sum\limits}
\newcommand{\convas}{\stackrel{a.s.}{\rightarrow}}

\newcommand{\calG}{\mathcal{G}}
\newcommand{\calF}{\mathcal{F}}
\newcommand{\calC}{\mathcal{C}}
\newcommand{\tL}{\tilde{L}}

\newcommand{\hN}{\hat{N}}
\newcommand{\HSIC}{\mbox{HSIC}}

\providecommand{\examplename}{Example}
\providecommand{\remarkname}{Remark}
\providecommand{\corollaryname}{Corollary}
\providecommand{\theoremname}{Theorem}
\providecommand{\lemmaname}{Lemma}

\global\long\def\b#1{{\bf \bm{\mathit{#1}}}}

\global\long\def\bA{\b A}

\global\long\def\bB{\b B}

\global\long\def\bH{\b H}

\global\long\def\bI{\b I}

\global\long\def\bJ{\b J}

\global\long\def\bL{\b L}

\global\long\def\bK{\b K}

\global\long\def\bx{\b x}

\global\long\def\by{\b y}

\global\long\def\bz{\b z}

\global\long\def\bGamma{\b{\Gamma}}

\global\long\def\bSigma{\b{\Sigma}}

\global\long\def\calX{\mathcal{X}}

\global\long\def\calT{\mathcal{T}}

\global\long\def\calL{\mathcal{L}}

\global\long\def\h#1{\hat{#1}}

\global\long\def\hd{\h d}

\global\long\def\halpha{\h{\alpha}}

\global\long\def\hbeta{\h{\beta}}

\global\long\def\t#1{\tilde{#1}}

\global\long\def\tK{\t K}

\global\long\def\tT{\t T}

\global\long\def\T{\top}

\global\long\def\tr{\operatorname{tr}}

\global\long\def\diag{\operatorname{diag}}

\global\long\def\E{\operatorname{E}}

\global\long\def\Var{\operatorname{Var}}

\global\long\def\Cov{\operatorname{Cov}}

\global\long\def\convp{\stackrel{p}{\longrightarrow}}

\global\long\def\convL{\stackrel{\calL}{\longrightarrow}}
\global\long\def\convas{\stackrel{a.s.}{\longrightarrow}}

\global\long\def\dequ{\stackrel{d}{=}}

\global\long\def\iidsim{\stackrel{\text{i.i.d.}}{\sim}}

\usepackage{xr}
\begin{document}

	\global\long\def\GCBL{\operatorname{GCBL}}

    \author[a]{\small Jin-Ting Zhang}
    \author[b]{\small Tianming Zhu}

    \affil[a]{\footnotesize Department of Statistics and Data Science, National University of Singapore, 
    	Singapore}
     \affil[b]{\footnotesize National Institute of Education, Nanyang Technological University, 
 	Singapore}

    \title{A fast and accurate kernel-based independence test with applications to  high-dimensional and  functional data}
	\maketitle
	
	\vspace{-1cm}

\begin{abstract}

 Testing the dependency between two random variables is an important  inference problem in statistics since many statistical procedures rely on the assumption  that the two samples are independent.
To test whether two samples  are independent, a so-called HSIC (Hilbert--Schmidt Independence Criterion)-based test    has been proposed. Its null distribution is approximated either
by permutation  or a Gamma approximation. In this paper, a new  HSIC-based  test is proposed.  Its asymptotic null and alternative distributions are  established. It is shown that the proposed test is
root-$n$ consistent.  A three-cumulant matched chi-squared approximation is adopted to approximate the null distribution of the test statistic.    By choosing a proper reproducing kernel, the proposed test can be applied to many different types of data including multivariate,  high-dimensional, and functional data. Three simulation studies and two real data applications show that in terms of level accuracy, power, and computational cost,  the proposed test outperforms several existing tests for multivariate, high-dimensional, and functional data.

\end{abstract}


\section{Introduction}\label{intro.sec}

With development of data collection techniques, complicated data objects such as high-dimensional data or functional data in  some separable metric spaces  are frequently encountered in various areas. In many big data applications nowadays, we are often interested in measuring the level of association between a pair of potentially high-dimensional random vectors or functional random variables.
 Testing the independence of random elements is an important inference  problem in statistics and has important applications. The work of this study is motivated by
 the Canadian weather data set,  available in the \textsf{R} package \textsf{fda.usc}  and discussed in details by  \cite{ramsay2005functional}. This Canadian weather data set has
 been studied in the literature of multivariate functional data analysis; see \cite{gorecki2017multivariate}  and \cite{zhu2022one}. It contains the average daily temperature curves and
 the average  daily precipitation curves  at $35$ Canadian  weather stations over a year, obtained via averaging the daily temperature curves and the daily precipitation curves yearly
 over the period $1960$ to $1994$.  Of interest is to check whether  the average  daily temperature curves and  the average  daily precipitation curves are statistically independent.  This is a two-sample independence testing  problem for functional data. If the above testing problem is rejected, we may take this dependence  into account in an inference procedure so that it is more efficient.

Mathematically, a two-sample independence testing  problem can be described as follows. Let $x$ and $y$ be two random elements defined in two separable metric spaces $\calX$ and $\calY$, respectively. Suppose we have a paired sample:
\be\label{twosamp.sec1}
(x_i, y_i),\; i=1,\ldots, n,
\ee
with each $(x_i, y_i)\in\calX \times \calY$ independently and identically following the joint  Borel  probability measure  $P_{xy}$. Of interest is to test the following hypotheses:
\be\label{hypo.sec2}
H_0: P_{xy}=P_xP_y, \;\mbox{ versus }\;   H_1: P_{xy}\neq P_xP_y,
\ee
where $P_x$ and $P_y$ be the marginal probability measures of $x$ and $y$, respectively.

There exist some classical dependence measures such as Spearman's $\rho$ and Kerdall's $\tau$ which have been widely applied. However, they are typically designed to capture only particular forms of dependence (e.g. linear or monotone) and they are not able to detect all modes of dependence between  random variables. As availability of complicated data objects, dependence measures are sought that capture more complex dependence patterns and those that occur between high-dimensional datasets or functional datasets.

Hilbert--Schmidt Independence Criterion (HSIC), introduced and studied by \cite{gretton2005measuring, Gretton05aJMLR}, is one of the most successful nonparametric dependence measures. It uses the distance between the kernel embeddings of probability measures in a reproducing kernel Hilbert space (RKHS) (\citealt{NIPS2007_d5cfead9, smola, Zhang2011KernelbasedCI}) which can be used for measuring the dependency between not only univariate or multivariate random variables, but also random variables valued into more complex structures such as  high-dimensional data or functional data. By employing HSIC, \cite{NIPS2007_d5cfead9} proposed  a novel test whose test statistic is an  empirical estimate of HSIC using V-statistics. The authors  approximated the  null distribution of the test statistic by  a two-parameter Gamma distribution. The resulting Gamma approximation based test costs $\mathcal{O}(n^2)$, where $n$ is the sample size. However, the simulation results in Tables~\ref{size1.tab} and~\ref{size2.tab} indicate that the Gamma approximation based  test  works well only  when the dimension $p$ is small and it is very conservative or totally fails to work when  the dimension $p$ is large. \cite{NIPS2007_d5cfead9} also proposed a permutation test  which works well generally in terms of size control but it  costs $\mathcal{O}(mn^2)$ operations, where $m$ is the number of permutations, indicating  that it is about $m$ times more time-consuming than the Gamma approximation based test. This is partially confirmed by Table~\ref{time1.tab} of  Section~\ref{simu1.sec} which shows that the permutation test  is about 100 to 1000  times more time-consuming than  the Gamma approximation based test. To overcome this problem, \cite{ZhangFilippiGretton2018StatComput} introduced three fast estimators of HSIC to speed up the computation in HSIC based tests. However, the computational complexity reduced by their fast estimators is not a free lunch. According to  the simulation results in \cite{ZhangFilippiGretton2018StatComput}, there is some loss in power when the sample size is small.  A much lower computational cost in large-scale examples is then offset by the requirement for a larger sample size.

In recent years, functional data analysis has emerged as an important area of statistics.  Most studies are conducted by assuming  that the random curves are independent without any checking. To overcome  this problem, a few methods have been developed  for detecting the dependency between random curves. Most of these independence tests are based on the measures of correlation including
the classical Pearson correlation (\citealt{pearson1895notes}), the dynamical correlation (\citealt{dubin2005dynamical}),  and the global temporal correlation (\citealt{zhou2018local}) among others.
 However, since zero correlation does not  imply independence generally, these functional correlations may be  insufficient for  independence testing (\citealt{miao2022wavelet}).
 \cite{kosorok2009discussion} applied the distance covariance proposed by \cite{szekely2007measuring} to top FPC scores which cumulatively account for 95\% of the variations of  random functions.
 Unfortunately, as discussed in \cite{shen2019distance} and the simulation results shown in Table~\ref{power3.tab}, for testing the dependence between two random functions, the correlation
 and distance covariance based  tests are less powerful  for non-monotone dependencies, although they are powerful  for monotone relationship.

In this paper, we propose a new HSIC-based test which works well  for multivariate,  high-dimensional, and functional data and it computes very fast. To the best of our knowledge, there are few tests which work well for multivariate,  high-dimensional, and functional data. The main contributions of this work are as follows. First of all, we propose an unbiased and root-$n$ consistent estimator for the centered reproducing kernel used in the proposed test statistic. It gives a good basis  for the proposed new test to have  much better size control than the Gamma approximation based test of \cite{NIPS2007_d5cfead9}.  Second, under some regularity conditions,  we show that under the null hypothesis, the proposed  test statistic has a chi-squared-type mixture limit.  Third, we derive the first three cumulants (mean, variance, and third central moment) of the proposed test statistic. This allows us to  employ  the three-cumulant (3-c) matched $\chi^2$-approximation of \cite{zhang2005approximate} to accurately approximate the distribution of  the chi-squared-type mixture with the approximation parameters consistently estimated from the data.  The 3-c matched $\chi^2$-approximation avoids permutation and significantly reduces the computational cost. It guarantees that the proposed new test computes very fast and has a good size control.  Fourth, we derive the asymptotic power of the proposed  new test under a local alternative and show that it is root-$n$ consistent. To the best of our knowledge, this has not been considered in the literature. Lastly, via three simulation studies and two real data examples, we demonstrate  that in terms of size control, power, and computational cost,  our new test works well and outperforms several existing tests for independence testing for multivariate, high-dimensional,  and functional data.

The rest of this paper is organized as follows. The main results are presented in Section~\ref{main.sec}. Simulation studies and real data applications are given in Sections~\ref{simu_st} and~\ref{real_data.sec}, respectively. Some concluding remarks are given in Section~\ref{con.sec}. Technical proofs of the main results are outlined in the Appendix.

\section{Main results} \label{main.sec}

\subsection{Test statistic}

Let $K(\cdot,\cdot): \calX\times\calX\rightarrow \calR$ and $L(\cdot,\cdot): \calY\times\calY\rightarrow \calR$ be two continuous,  positive characteristic  reproducing  kernels. Let $\calF$ and $\calG$ be the two reproducing kernel Hilbert spaces (RKHS) with inner products $\langle,\rangle_{\calF}$ and $\langle,\rangle_{\calG}$,  generated by $K$ and $L$, respectively.
 Let
 $\phi(x)=K(x,\cdot)$ and $\psi(y)=L(y,\cdot)$
  denote their associated canonical feature maps such that we have the following kernel tricks:
\[
K(x,x')=\langle\phi(x),\phi(x')\rangle_{\calF},\; \mbox{ and }\; L(y,y')=\langle\psi(y),\psi(y')\rangle_{\calG},
\]
where $(x',y')$ is an independent copy of $(x,y)$. It follows that $\phi(\calX)\subset \calF$ and $\psi(\calY)\subset\calG$.
Following \cite{fukumizu2004dimensionality}, the cross-covariance operator $\calC_{xy}:\calG\to\calF$ is defined such that for all $f\in\calF$ and $g\in\calG$
\be\label{fgCov.sec2}
\langle f, \calC_{xy}g\rangle_{\calF} =\Cov[f(x),g(y)].
\ee
Set $\mu_{x}=\E_x[\phi(x)]$ and  $\mu_y=\E_y[\psi(y)]$ to be the mean embeddings of the probability measures $P_x$ and $P_y$, respectively. The cross-covariance operator itself can then be written
\[
\calC_{xy}=\E_{xy}\{[\phi(x)-\mu_x]\otimes[\psi(y)-\mu_y]\},
\]
where $\otimes$ denotes the tensor product.

According to \citet[Theorem 4]{Gretton05aJMLR}, $x$ and $y$ are independent if and only if $\Cov[f(x),g(y)]=0$ for all continuous, bounded functions $f(x)$ and $g(y)$. Then the expression (\ref{fgCov.sec2}) indicates that $x$ and $y$ are independent if and only if $\calC_{xy}=0$. That is to say,  $x$ and $y$ are independent if and only if their kernel embeddings  $\phi(x)$ and $\psi(y)$ are uncorrelated. Therefore,  to test (\ref{hypo.sec2}) using the original i.i.d. sample (\ref{twosamp.sec1}) is equivalent to test the following hypotheses:
\begin{equation}
\label{ihypo.sec2}
H_0: \calC_{xy}=0, \text{ versus }  H_1: \calC_{xy}\neq 0.
\end{equation}
\cite{NIPS2007_d5cfead9} proposed to measure the dependence between $x$ and $y$ using the following squared  Hilbert--Schmidt-norm of $\calC_{xy}$:
\be\label{HSnorm.sec2}
\begin{array}{rcl}
\|\calC_{xy}\|^2_{\calF\otimes\calG}	&=&\E_{xyx'y'}\Big[\langle\phi(x) - \mu_x,\phi(x') - \mu_{x'}\rangle_{\calF}\langle\psi(y) - \mu_y,\psi(y')-\mu_{y'}\rangle_{\calG}\Big]\\
&=&\E_{xyx'y'}[\tK(x,x')\tL(y,y')],
\end{array}
\ee
where   $\tK(x,x')$ and  $\tL(y,y')$  denote the centered versions of $K(x,x')$ and $L(y,y')$, respectively, and $(x',y')$ is an independent copy of $(x,y)$. Notice that for  the kernel $K(\cdot,\cdot)$,  we have
\be\label{tK.sec2}
\begin{array}{rcl}
	\tK(x,x')&=&\langle\phi(x) - \mu_x,\phi(x') - \mu_{x'}\rangle_{\calF}\\
             &=&K(x,x')-\E_{z'}[K(x,z')]-\E_{z}[K(z,x')]+\E_{z,z'}[K(z,z')],
\end{array}
\ee
where  $z$ and $z'$ are independent copies of $x$ and $x'$, respectively. Notice also  that we have the following useful properties: when $x'=x$, we have
$$\E[\tK(x,x)]=\langle\phi(x) - \mu_x,\phi(x) - \mu_{x}\rangle_{\calF}=\E\|\phi(x)-\mu_x\|_{\calF}^2\ge 0,$$
 and when  $x$ and $x'$ are independent, we have
$$\E_{x}[\tK(x,x')]=\E_{x'}[\tK(x,x')]=\E_{x,x'}[\tK(x,x')]=0.$$
The above properties are valid after replacing $K$ and $x$ with $L$ and $y$, respectively.
Using (\ref{HSnorm.sec2}), to test (\ref{ihypo.sec2}), we can construct the following test statistic
$$T_n=n^{-1}\sum_{i=1}^n\sum_{j=1}^n\tK^*(x_i,x_j)\tL^*(y_i,y_j),$$
where $\tK^*(x_i,x_j)$ and $\tL^*(y_i,y_j)$ are the unbiased estimators of $\tK(x_i,x_j)$ and $\tL(y_i,y_j)$, respectively, which are given by
		\be\label{tKstar.sec2}
		\begin{split}
		\tK^*(x_i,x_j)&=K(x_i,x_j)-\sum_{v\neq i}\frac{K(x_i,x_v)}{n-1}-\sum_{u\neq j}\frac{K(x_u,x_j)}{n-1}+\sum_{u\neq v}\frac{K(x_u,x_v)}{n(n-1)},\;\mbox{ and }\\
		\tL^*(y_i,y_j)&=L(y_i,y_j)-\sum_{v\neq i}\frac{L(y_i,y_v)}{n-1}-\sum_{u\neq j}\frac{L(y_u,y_j)}{n-1}+\sum_{u\neq v}\frac{L(y_u,y_v)}{n(n-1)}.
		\end{split}
		\ee
For simplicity, let $\bK=(K(x_i,x_j)):n\times n$ and  $\bL=(L(y_i,y_j)):n\times n$ denote the Gram matrices of the two kernels   $K(\cdot,\cdot)$ and $L(\cdot,\cdot)$, respectively. Similarly,
set $\tbK=(\tK(x_i,x_j)):n\times n$ and  $\tbL=(\tL(y_i,y_j)):n\times n$, and $\tbK^*=(\tK^*(x_i,x_j)):n\times n$ and  $\tbL^*=(\tL^*(y_i,y_j)):n\times n$. 
Then we have
\be\label{stat1.sec2}
T_n=n^{-1}\tr(\tbK^*\tbL^*),
\ee
where $\tr(\bA)$ denotes the trace of the square matrix $\bA$, i.e., the sum of the diagonal entries of $\bA$. Note that $T_n$  can be easily computed using $\mathcal{O}(n^2)$ operations.

\begin{rem}\label{Gstat0.rem} To test (\ref{ihypo.sec2}), \cite{NIPS2007_d5cfead9}  proposed the following  test statistic:
	\be\label{Gstat0.sec2}
	T_{n,G}=n^{-1}\tr(\bH\bK\bH\bL)=n^{-1}\tr(\bK_G\bL_G)=n^{-1}\sum\limits_{i=1}^n\sum\limits_{j=1}^n K_G(x_i,x_j)L_G(y_i,y_j),
	\ee
	where $\bH=\bI_n-\bJ_n/n$  with $\bI_n$ and $\bJ_n$ being the $n\times n$ identical matrix and the $n\times n$ matrix of ones, and $\bK_G=(K_G(x_i,x_j)):n\times n$ and  $\bL_G=(L_G(y_i,y_j)):n\times n$ with
	 $K_G(x_i,x_j)$ and $L_G(y_i,y_j)$ being  the biased estimators of  $\tK(x_i,x_j)$ and $\tL(y_i,y_j)$ defined as
\be\label{GKstar.sec2}
	\begin{array}{rcl}
	K_G(x_i,x_j)&=&K(x_i,x_j)-\sum\limits_{u=1}^n  \frac{K(x_i, x_{u})}{n} -\sum\limits_{v=1}^n  \frac{K(x_{v}, x_j)}{n}+\sum\limits_{v=1}^n \sum\limits_{u=1}^n \frac{K(x_{v}, x_{u})}{n^2},\;\mbox{ and }\\
	L_G(y_i,y_j)&=&L(y_i,y_j)-\sum\limits_{u=1}^n  \frac{L(y_i, y_{u})}{n} -\sum\limits_{v=1}^n  \frac{L(y_{v}, y_j)}{n}+\sum\limits_{v=1}^n \sum\limits_{u=1}^n \frac{L(y_{v}, y_{u})}{n^2}.
	\end{array}
	\ee
 It is worthwhile to emphasize that the differences between  our test statistic $T_n$ (\ref{stat1.sec2}) and \cite{NIPS2007_d5cfead9}'s test statistic $T_{n,G}$ (\ref{Gstat0.sec2}) mainly come from the differences between the unbiased estimators (\ref{tKstar.sec2}) and the biased estimators (\ref{GKstar.sec2}) of $\tK(x_i,x_j)$ and $\tL(y_i,y_j)$.
\end{rem}

\subsection{Asymptotic null distribution}

Assume  that
\be\label{cond.sec2}
|K(x,x')|<B_K<\infty, \; \mbox{ and } \;|L(y,y')|<B_L<\infty  \mbox { for  all $x,x'\in\calX$ and $y,y'\in\calY$},
\ee
where $B_K$ and $B_L$ are two constants. Condition (\ref{cond.sec2}) guarantees that $|\tK(x,x')|\le 4B_K$, $|\tL(y,y')|\le 4B_L$, $\E[\tK(x,x)]<\infty$, and $\E[\tL(y,y)]<\infty$.
Then by the Cauchy--Schwarz inequality, both $\tK(x,x')$ and $\tL(y,y')$  are square integrable, i.e., $\E[\tK^2(x,x')]<\infty$ and $\E[\tL^2(y,y')]<\infty$. Thus $\tK(x,x')$ and $\tL(y,y')$  have the following Mercer's expansions
\be\label{mercer.sec2}
\tK(x,x')=\sum\limits_{r=1}^{\infty} \lambda_r \phi_r(x)\phi_r(x'),\; \mbox{ and } \; \tL(y,y')=\sum\limits_{r=1}^{\infty} \rho_r \psi_r(y)\psi_r(y'),
\ee
where $\lambda_1, \lambda_2, \dots $ and $\rho_1, \rho_2, \dots $ are the eigenvalues of $\tK(x,x')$ and $\tL(y,y')$, respectively, and with a slight abuse of notation,  $\phi_1(x),\phi_2(x),\dots$ and $\psi_1(y),\psi_2(y),\dots$ are the associated orthonormal eigen-elements  of $\tK(x,x')$ and $\tL(y,y')$, respectively in the sense that
\[
\begin{array}{rcl}
\int_{\calX} \tK(x,x') \phi_r(x) P_x(dx) =\lambda_r \phi_r(x'), &&\;\;  \int_{\calX}\phi_{r}(x) \phi_{s}(x) P_x(dx)=\delta_{rs}, \; r,s =1,2,\dots, \;\\
\int_{\calY} \tL(y,y') \psi_r(y) P_y(dy) =\rho_r \psi_r(y'),&& \;\;  \int_{\calY}\psi_{r}(y) \psi_{s}(y) P_y(dy)=\delta_{rs}, \; r,s =1,2,\dots, \;\\
\end{array}
\]
where $\delta_{rs}=1$ when $r=s$ and $0$ otherwise.
To derive the asymptotic null distribution of $T_n$, we need the following lemma whose proof  is given in the Appendix.
\begin{lem}\label{lem1.sec2}
Under the condition (\ref{cond.sec2}),  as $n\to\infty$,  we have
\[
\tK^*(x_i,x_j)=\tK(x_i,x_j)+\mathcal{O}(n^{-1/2})\; \mbox{uniformly for all } x_i,x_j.
\]
\end{lem}
Lemma~\ref{lem1.sec2} gives the uniform convergence rate of $\tK^*(x_i,x_j)$ to $\tK(x_i,x_j)$. Similarly, we can also have $\tL^*(y_i,y_j)=\tL(y_i,y_j)+\mathcal{O}(n^{-1/2})$ uniformly for all $y_i,y_j$. Therefore, we can write that
	\be\label{tTn2.sec2}
	T_n=n^{-1}\tr(\tbK\tbL)+\mathcal{O}(n^{-1/2})\equiv \tT_n+\mathcal{O}(n^{-1/2}),
	\ee
	 That is, $T_n$ and $\tT_n$ have the same distribution for large values of $n$. Thus, studying the asymptotic null distribution of $T_n$ is equivalent to studying that of $\tT_n$.

\begin{thm} \label{tTndist2.thm} Under the condition (\ref{cond.sec2}) and the null hypothesis, as $n\to\infty$, we have   $\tT_n \convL \tT$ and $T_n\convL\tT$ with
$$\tT \dequ  \sum\limits_{r=1}^{\infty} \sum\limits_{s=1}^{\infty} \lambda_r\rho_s  A_{rs}, \;  A_{rs} \stackrel{i.i.d.}{\sim} \chi_{1}^2,$$
where $\lambda_1, \lambda_2, \dots $ and $\rho_1, \rho_2, \dots $ are the eigenvalues of $\tK(x,x')$ and $\tL(y,y')$, respectively.
\end{thm}

\begin{rem}\label{Gstat1.rem} Theorem~\ref{tTndist2.thm} is parallel to Theorem 2 of \cite{NIPS2007_d5cfead9} where the authors treated $T_{n,G}/n$ (See Remark~\ref{Gstat0.rem}) as a V-statistic of an order $4$ kernel  while actually we can show that  $T_{n,G}/n=\tT_n/n+\mathcal{O}(n^{-3/2})$ where   $\tT_n/n$ is a V-statistic of an order $2$ kernel  only [see (\ref{tTn2.sec2}) for details]. Theorem~\ref{tTndist2.thm} is the same as Theorem 1 of \citet{ZhangFilippiGretton2018StatComput} but our proof is much simpler than that of the latter.
\end{rem}

\subsection{Null distribution approximation \label{implement.sec}}

As mentioned in the introduction section, \cite{NIPS2007_d5cfead9} approximated the  null distribution of $T_{n,G}$ (\ref{Gstat0.sec2})  by permutation and by a two-parameter Gamma distribution, resulting in a permutation test and a Gamma approximation based test. In this subsection, since $\tT$  is  a $\chi^2$-type mixture with unknown coefficients,  we  approximate the null distribution of $T_n$  using the three-cumulant (3-c) matched $\chi^2$-approximation (\citealt{zhang2005approximate}, \citealt{zhang2013analysis}), resulting in a 3-c matched $\chi^2$-approximation based new test.

 \begin{rem}\label{Gapprox.rem}
 In terms of computational cost,  the permutation test is  very  time-consuming with a cost $\mathcal{O}(mn^2)$ where $m$ is the number of permutations and $n$ is the sample size while the Gamma approximation based test computes very fast, with a cost of $\mathcal{O}(n^2)$. However, in terms of size control, the permutation test generally performs quite well but  the Gamma approximation based test  performs well only for low dimensional data and it is very conservative and totally fails to work for high-dimensional data,   as demonstrated by the simulation results presented in Tables~\ref{size1.tab} and~\ref{size2.tab} of  Section~\ref{simu1.sec}.
 \end{rem}

The key idea of the proposed new test  is to approximate the null distribution of $T_n$ using that of the following random variable of form
$$R\dequ \beta_0+\beta_1\chi_d^2.$$
 The parameters $\beta_0, \beta_1$, and $d$ are determined via matching the first three cumulants of $T_n$ and $R$. For this purpose, we  derive the first three cumulants of $T_n$ as in the following theorem whose proof  is given in the Appendix.

\begin{thm}\label{muvar2.thm} Under the condition (\ref{cond.sec2}) and the  null hypothesis, the first three cumulants of $T_n$ are given by
\be\label{tT3c.sec2}
\begin{array}{c}
\E(T_n)=M_1N_1+\mathcal{O}(n^{-1/2}),\;\;
\Var(T_n)=2M_2N_2+\mathcal{O}(n^{-1/2}), \; \mbox{ and }\\
\E[T_n-\E(T_n)]^3=8M_3N_3+\mathcal{O}(n^{-1/2}),
\end{array}
\ee
where with $x,x',x''\iidsim P_x$ and $y,y',y''\iidsim P_y$,
  \be\label{MNexp.sec2}
  \begin{array}{llll}
   &M_1=\E[\tK(x,x)],\; &M_2=\E[\tK^2(x,x')], \; &M_3=\E[\tK(x,x')\tK(x',x'')\tK(x'',x)],\\
    &N_1=\E[\tL(y,y)],\; &N_2=\E[\tL^2(y,y')], \; &N_3=\E[\tL(y,y')\tL(y',y'')\tL(y'',y)].
    \end{array}
 \ee
\end{thm}


 The first three cumulants of $R$ are given by  $\beta_0+\beta_1 d$, $2\beta_1^2 d$, and $8\beta_1^3 d$ while the first three cumulants of $T_n$ are given in (\ref{tT3c.sec2}).
  Equating  the first three-cumulants of $T_n$ and $R$ and ignoring the higher order terms  then leads to
  \be\label{betadf.sec2}
\beta_0=M_1N_1-\frac{(M_2N_2)^2}{M_3N_3},\;\;
\beta_1=\frac{M_3N_3}{M_2N_2},\;\;\mbox{ and }\;
d=\frac{(M_2N_2)^3}{(M_3N_3)^2}.
\ee
 In addition,  the skewness of $T_n$ can also  be approximately expressed as
$$\frac{\E[T_n-\E(T_n)]^3}{\Var^{3/2}(T_n)}=\frac{8M_3N_3}{(2M_2N_2)^{3/2}}=\sqrt{8/d}.$$
Thus the skewness of $T_n$  will become small as $d$ increases.

Let $\hM_1, \hM_{2}, \hM_3$ and $\hN_1,\hN_2,\hN_{3}$  be the consistent estimators of $M_1, M_{2}, M_3$ and $N_1, N_2, N_3$ respectively.  Plugging  them into (\ref{betadf.sec2}), the consistent estimators of $\beta_0,\beta_1$ and $d$ are then obtained as
\be\label{hbetadf.sec2}
\hbeta_0=\hM_1\hN_1-\frac{(\hM_2\hN_2)^2}{\hM_3\hN_3},\;\;
\hbeta_1=\frac{\hM_3\hN_3}{\hM_2\hN_2},\;\mbox{ and }\;
\hd=\frac{(\hM_2\hN_2)^3}{(\hM_3\hN_3)^2}.
\ee
Then for any nominal significance level $\alpha>0$, let
$\chi_{d}^2(\alpha)$ denote the upper $100\alpha$ percentile of
$\chi_{d}^2$. Then using  (\ref{hbetadf.sec2}),  the proposed new test  with   the  3-c matched  $\chi^2$-approximation can   then be  conducted via using  the approximate critical
value $\hbeta_0+\hbeta_1\chi_{\hd}^2(\alpha)$ or  the approximate $p$-value
$P\{\chi_{\hd}^2\ge (T_n-\hbeta_0)/\hbeta_1\}$.

Under the condition (\ref{cond.sec2}), by Lemma~\ref{lem1.sec2}, we have $\tK^*(x_i,x_j)=\tK(x_i,x_j)+\mathcal{O}(n^{-1/2})$ uniformly for all $x_i,x_j$'s, and  $\tL^*(y_i,y_j)=\tL(y_i,y_j)+\mathcal{O}(n^{-1/2})$ uniformly for all $y_i,y_j$'s.
	Then
by (\ref{MNexp.sec2}),  the natural  estimators of $M_1,M_2,$ $M_3$ and $N_1,N_2,$ $N_3$ are given by
\be\label{hMNexp1.sec2}
\begin{array}{c}
\hM_1=\frac{1}{n}\sum\limits_{i=1}^n \tK^*(x_i,x_i),\;\;\;
\hM_2=\frac{2}{n(n-1)} \sum\limits_{1\le i<j\le n} [\tK^*(x_i,x_j)]^2,\\
\hM_3=\frac{6}{n(n-1)(n-2)} \sum\limits_{1\le i<j<k\le n}[\tK^*(x_i,x_j)\tK^*(x_j,x_k)\tK^*(x_k,x_i)],
\end{array}
\ee
and
\[
\begin{array}{c}
\hN_1=\frac{1}{n}\sum\limits_{i=1}^n \tL^*(y_i,y_i),\;\;
\hN_2=\frac{2}{n(n-1)} \sum\limits_{1\le i<j\le n} [\tL^*(y_i,y_j)]^2,\\
\hN_3=\frac{6}{n(n-1)(n-2)} \sum\limits_{1\le i<j<k\le n}[\tL^*(y_i,y_j)\tL^*(y_j,y_k)\tL^*(y_k,y_i)],
\end{array}
\]
where $\tK^*(x_i,x_j)$ and $\tL^*(y_i,y_j)$ are defined in (\ref{tKstar.sec2}). For fast computation, we can write $\hM_1=\tr(\tbK^*)/n$,
\be\label{M23.sec2}
\begin{array}{rcl}
\hM_2&=&[n(n-1)]^{-1}\left[\tr(\tbK^{*2})-\tr(\tbK^*o\tbK^*)\right],\;\mbox{ and }\\
\hM_3&=&[n(n-1)(n-2)]^{-1}\left\{\tr(\tbK^{*3})-3\tr[\diag(\tbK^*)\tbK^{*2}]+2\tr(\tbK^*o \tbK^* o\tbK^*)\right\},
\end{array}
\ee
where $\bA o\bB=(a_{ij}b_{ij})$  denotes a dot product of two matrices $\bA=(a_{ij})$ and $\bB=(b_{ij})$, and $\diag(\bA)$ denotes a diagonal matrix formed by the diagonal entries of $\bA$.
 The proof of (\ref{M23.sec2}) is given in the Appendix.  Similarly, we have $\hN_1=\tr(\tbL^*)/n$, $\hN_2=[n(n-1)]^{-1}[\tr(\tbL^{*2})-\tr(\tbL^*o\tbL^*)]$, and
\[
\hN_3=[n(n-1)(n-2)]^{-1}\left\{\tr(\tbL^{*3})-3\tr[\diag(\tbL^*)\tbL^{*2}]+2\tr(\tbL^*o\tbL^*o\tbL^*)\right\}.
\]
\begin{thm}\label{hcoefs.thm} Under the condition (\ref{cond.sec2}), as $n\to \infty$, we have $\hat{M}_{\ell} \stackrel{p}{\longrightarrow} M_{\ell}, \ell=1,2,3$ and $\hat{N}_{\ell} \stackrel{p}{\longrightarrow} N_{\ell}, \ell=1,2,3$. It follows that as $n\to \infty$, we have
$
\hat{\beta}_0 \stackrel{p}{\longrightarrow} \beta_0, \hat{\beta}_1\stackrel{p}{\longrightarrow} \beta_1$, and $\hat{d} \stackrel{p}{\longrightarrow} d.$
\end{thm}

\begin{rem}\label{GapproxT.rem} Since the Gamma approximation based test matches the mean and variance of $T_{n,G}$ while the proposed new test matches the mean, variance, and the third central moment of $T_n$, it is expected that in terms of size control,  the proposed new test should  outperform the Gamma approximation based test substantially. This is partially confirmed by the simulation results presented in Tables~\ref{size1.tab} and ~\ref{size2.tab} of  Section~\ref{simu1.sec}. Further, in terms of computational cost,  the proposed new test, with a cost of $\mathcal{O}(n^3)$,  is much less time-consuming than the permutation test when the number of permutations is larger than the sample size  and is slightly more time-consuming than the Gamma approximation based test. This is partially  confirmed by Table~\ref{time1.tab} of  Section~\ref{simu1.sec}.
\end{rem}

\subsection{Asymptotic power}

In this subsection, we investigate the asymptotic power of the proposed test under the following local alternative hypothesis:
\be\label{localt.sec2}
H_{1n}: \HSIC=\|\calC_{uv}\|_{\calF\otimes \calG}=n^{-(1/2-\Delta)}h,
\ee
where $0<\Delta<1/2$ and $h$ is a positive constant.  The above local alternative hypothesis   will tend to the null hypothesis as the sample size $n$  tends to infinity. Therefore, it is  often  challenging to detect it. A test is usually  called  to be root-$n$ consistent if it can detect the local alternative hypothesis (\ref{localt.sec2})  with probability tending to $1$ as $n$ tends to infinity. A root-$n$ consistent test is often preferred.

\begin{thm}\label{power.thm} Under the condition (\ref{cond.sec2}) and the local alternative (\ref{localt.sec2}), as $n\to\infty$, we have
\[
\sqrt{n}\left(T_n/n-\HSIC\right)\convL N(0,4\sigma^2),
\]
where with $(x',y')$ being an independent copy of $(x,y)$, $\sigma^2=\Var\left\{\E[\tK(x,x')\tL(y,y')|(x',y')]\right\}$. It follows that for any significance level $\alpha$, the asymptotic power of the proposed test $T_n$ is given by
\[
Pr\left[T_n\ge \hbeta_0+\hbeta_1\chi_{\hd}^2(\alpha)\right]=\Phi\left[n^{\Delta}h/(2\sigma) \right][1+o(1)],
\]
which tends to $1$ as $n\to\infty$ where $\Phi(\cdot)$ denotes the cumulative distribution function of $N(0,1)$.
\end{thm}

 The proof  of Theorem~\ref{power.thm} is  given in the Appendix.  Theorem~\ref{power.thm} shows that the proposed new test $T_n$ is root-$n$ consistent.

\begin{rem}\label{Gpower.rem}
The first result of Theorem~\ref{power.thm} is parallel  to Theorem 1 of \cite{NIPS2007_d5cfead9} where the authors treated $T_{n,G}/n$ as a V-statistic of an order $4$ kernel  while actually we have $T_{n,G}/n=\tT_n/n+O_p(n^{-3/2})$ with $\tT_n/n$ being a V-statistic of an order $2$ kernel  only [see (\ref{tTn2.sec2}) for details]. The result of Theorem 1 of \cite{NIPS2007_d5cfead9} may be  problematic.
\end{rem}

\section{Simulation studies}\label{simu_st}

In this section, we conduct three simulation studies, namely Simulations 1, 2, and 3,  to compare  the proposed new test, denoted as NEW, against several existing competitors for the two-sample independence testing problem for multivariate, high-dimensional, and functional data, respectively. We compute the empirical size or power of a test as the proportion of the number of rejections out of $10,000$ simulation runs. Throughout this section, we set the nominal size $\alpha$ as $5\%$.

In the  three simulation studies described below, for simplicity,  we choose the kernel $K(\cdot,\cdot)$ to be the following  Gaussian radial basis function (RBF) kernel:
$$K(x,x')=\exp\left(-\frac{\|x-x'\|^2}{2\sigma^2}\right),$$
where $\sigma^2$ is the so-called  kernel width. For multivariate and high-dimensional data as in Simulations 1 and 2, $\|x\|$ denotes the usual $L^2$-norm of a vector $x$ and for functional data as in Simulation 3, it denotes  the usual $L^2$-norm of a function $x(t),t\in\calT$ given by $\|x\|=\left[\int_{\calT}x^2(t)dt\right]^{1/2}$ and it is computed via approximating  the integrals using the trapezoidal rule. It is easy  to see that  the above Gaussian RBF kernel is bounded above by $1$ so that the condition (\ref{cond.sec2}) is always satisfied.  The kernel width $\sigma^2$ is selected by employing the data-adaptive Gaussian kernel width selection method proposed in \citet[sec. 2.6]{ZhangGuoZhou2022JOE}. For the kernel $L(\cdot,\cdot)$, it is  done similarly.

\subsection{Simulation 1}\label{simu1.sec}

In this simulation study, we demonstrate  the performance of the NEW test for multivariate data against  the Gamma approximation based test and the permutation test proposed and studied
 in \cite{NIPS2007_d5cfead9}, denoted as HSICg and HSICp, respectively. The HSICg test is implemented in the \textsf{R} package \textsf{dHSIC} (\citealt{pfister2017dhsic}) and the number of permutations used in the HSICp test is set as 200.

\begin{table}[!h]
	\caption{Empirical sizes and powers (in $\%$) of  the HSICp, HSICg, and NEW tests in  Simulation 1.}\label{size1.tab}	
	\begin{center}
			\begin{tabular}{ccccccccccc}
				\hline
				\multicolumn{2}{c}{}      & \multicolumn{3}{c}{$\theta=0$} & \multicolumn{3}{c}{$\theta=\pi/8$} & \multicolumn{3}{c}{$\theta=\pi/4$} \\
				\hline
				$n$     & $p$     & HSICp    & HSICg  & NEW    & HSICp    & HSICg  & NEW    & HSICp    & HSICg  & NEW\\
				\hline
		\multirow{4}[0]{*}{30} & 2     & 5.16  & 5.31  & 5.57  & 8.42  & 8.74  & 9.20  & 13.50 & 13.69 & 14.43 \\
		& 4     & 5.01  & 3.49  & 5.49  & 6.14  & 4.12  & 6.85  & 7.44  & 5.19  & 8.24 \\
		& 10    & 5.13  & 0.10  & 5.99  & 5.76  & 0.06  & 6.54  & 5.61  & 0.14  & 6.50 \\
		& 20    & 5.08  & 0.00  & 6.11  & 4.83  & 0.00  & 5.98  & 5.43  & 0.00  & 6.48 \\
		\hline
		\multirow{4}[0]{*}{50} & 2     & 5.08  & 5.39  & 5.32  & 11.84 & 12.34 & 12.32 & 21.66 & 22.58 & 22.52 \\
		& 4     & 5.22  & 4.34  & 5.42  & 6.67  & 5.64  & 7.20  & 9.59  & 8.05  & 10.08 \\
		& 10    & 5.33  & 0.40  & 5.78  & 5.52  & 0.35  & 5.95  & 5.47  & 0.48  & 5.75 \\
		& 20    & 4.86  & 0.00  & 5.41  & 4.77  & 0.00  & 5.60  & 5.44  & 0.00  & 6.12 \\
		\hline
		\multirow{4}[0]{*}{100} & 2     & 5.03  & 5.40  & 4.97  & 20.46 & 21.58 & 20.74 & 38.97 & 40.13 & 39.38 \\
		& 4     & 4.86  & 4.59  & 4.91  & 8.72  & 8.59  & 9.24  & 15.80 & 15.32 & 16.27 \\
		& 10    & 5.33  & 1.63  & 5.51  & 5.92  & 1.98  & 6.29  & 6.90  & 2.43  & 7.15 \\
		& 20    & 4.98  & 0.00  & 5.28  & 5.25  & 0.00  & 5.69  & 5.44  & 0.00  & 5.72 \\
		\hline
		\multirow{4}[0]{*}{200} & 2     & 5.12  & 5.54  & 5.18  & 38.11 & 39.78 & 38.76 & 58.50 & 59.83 & 59.07 \\
		& 4     & 5.03  & 5.27  & 5.22  & 15.44 & 15.99 & 15.86 & 30.76 & 31.25 & 31.17 \\
		& 10    & 5.19  & 3.10  & 5.16  & 6.53  & 3.93  & 6.73  & 9.58  & 6.27  & 9.68 \\
		& 20    & 4.66  & 0.03  & 4.88  & 5.45  & 0.00  & 5.49  & 6.35  & 0.12  & 6.50 \\
		\hline
		\multirow{4}[0]{*}{500} & 2     & 5.19  & 5.66  & 5.23  & 64.05 & 65.17 & 64.33 & 80.60 & 81.26 & 80.77 \\
		& 4     & 5.03  & 5.55  & 5.31  & 36.04 & 37.00 & 36.36 & 54.99 & 56.01 & 55.45 \\
		& 10    & 5.25  & 4.39  & 5.05  & 10.62 & 9.30  & 10.73 & 20.11 & 18.46 & 20.47 \\
		& 20    & 4.64  & 0.81  & 4.67  & 6.15  & 1.31  & 6.28  & 8.78  & 1.93  & 8.77 \\
		\hline
\end{tabular}
\end{center}
\end{table}

 We make  use of the multivariate benchmark data scheme used in  \cite{NIPS2007_d5cfead9}. We conduct  the independence test  for  $p$-dimensional random variables for  $p=2,4,10$, and $20$. The data are generated  as follows. First, using \textsf{rjordan} in the \textsf{R} package \textsf{ProDenICA} (\citealt{hastie2022package}),  we generate $n$ observations of two univariate random variables randomly and with replacement, each drawn at random from the Independent Component Analysis (ICA)  benchmark densities in Table 3 of \cite{Gretton05aJMLR}, including  super-Gaussian, sub-Gaussian, multimodal, and unimodal distributions.  Second, we mix these random variables using a rotation matrix parameterized by an angle $\theta$, varying from $0$ to $\pi/4$ (a zero angle means the data are independent, while dependence becomes easier to detect as the angle increases to $\pi/4$). That is, we set $\theta=0, \pi/8$, and $\pi/4$. Third, we append $p-1$ dimensional Gaussian noise of $0$ mean and $1$ standard deviation to each of the mixtures. Finally, we multiply each resulting vector by an independent random $p$-dimensional orthogonal matrix, to obtain vectors which are  dependent across all observed dimensions. The sample size we consider includes  $n=30,50,100,200$ and $500$, respectively.

The empirical sizes and powers of the HSICp, HSICg, and NEW  tests in Simulation 1 are displayed in Table~\ref{size1.tab}. We can draw several interesting  conclusions. When $\theta=0$, the null hypothesis holds so that we can compare the performances of the three tests in terms of size control.  It is seen that  in terms of size control, both the  HSICp and NEW tests  have very good level accuracy and their performances are generally comparable since  the empirical sizes of the  HSICp and NEW tests are generally around 5\% and below 6\% under most of the settings although the empirical sizes of the NEW test are slightly larger than those of the HSICp test. Admittedly, the HSICp test slightly outperforms the NEW test since the empirical sizes of the HSICp test range from $4.64\%$ to $5.33\%$ while the empirical sizes of the NEW test range from $4.67\%$ to $6.11\%$. However, the HSICg test performs much worse than  the HSICp and NEW tests. When $p=2$, the HSICg test performs quite well with its empirical sizes  generally around 5\%.  However,  with increasing the dimension $p$, the HSICg test becomes more and more conservative with its empirical sizes becoming as small as $0.00\%$ especially when $p\ge 10$ and the sample size $n$ is small. This means that the HSICg test does not work  for moderate or  high dimensional data while the HSICp and NEW tests still  work well.    On the other hand, when $\theta>0$, the alternative hypotheses hold so that we can compare the performances of the three tests in terms of power. As expected, the empirical powers of the HSICp and NEW tests are generally comparable although the empirical powers of the HSICp test are slightly smaller than those of the NEW test but when the dimension $p\ge 10$,  the empirical powers of the HSICg test are generally  smaller than those of the HSICp and NEW tests, showing the impact of the level accuracy of the three tests. Notice that  as the value of $\theta$ increases or as the sample size $n$ increases or both,   the empirical powers  of the three tests are generally getting larger.  Notice also that  as the dimension $p$ increases,  the empirical powers of each of the three tests are getting smaller.  This is  not surprising, however,  because  when $\theta\neq 0$,  only the first elements in the two variables are correlated so that as the dimension  $p$ increases, the dependence between the variables  becomes harder and harder  to detect.

In the above, we compare the HSICp, HSICg, and NEW  tests in terms of size control and power. We now compare their computational costs. To this end,  the total execution time (in minutes) of  the three tests  for the $10,000$ simulation runs when $\theta=0$,  $p=4,10$, and $20$ and $n=30,50,100,200$, and $500$  are presented in Table~\ref{time1.tab}.  It is seen that   the HSICp test is  $10\sim 100$ times  more time-consuming than the NEW  test although here the number of permutations is only $200$ while the NEW test is
only about $1\sim 10$ times more time-consuming than the HSICg test for the sample size $n=30,50,100,200$ and $500$.
\begin{table}[!h]
	\caption{Computational costs (in minutes) of the HSICp, HSICg, and NEW  tests in  Simulation 1 when $\theta=0$ and $p=4,10,$ and $20$ for $N=10,000$ simulation runs.}\label{time1.tab}	
	\begin{center}
		\begin{tabular}{cccccccccc}
			\hline
			& \multicolumn{3}{c}{$p=4$} & \multicolumn{3}{c}{$p=10$} & \multicolumn{3}{c}{$p=20$} \\
			\hline
			$n$     & HSICp    & HSICg  & NEW      & HSICp    & HSICg  & NEW      & HSICp    & HSICg  & NEW\\	
			\hline
			30    & 12.75 & 0.10  & 0.09  & 14.61 & 0.11  & 0.10  & 29.41 & 0.17  & 0.18 \\
			50    & 40.97 & 0.18  & 0.29  & 46.12 & 0.20  & 0.32  & 86.92 & 0.36  & 0.52 \\
			100   & 202.43 & 0.55  & 1.51  & 222.77 & 0.63  & 1.61  & 383.34 & 1.19  & 2.41 \\
			200   & 1145.73 & 1.96  & 9.32  & 1227.69 & 2.29  & 9.72  & 1865.75 & 4.55  & 12.89 \\
			500   & 13842.41 & 12.26 & 123.01 & 14395.42 & 14.45 & 125.72 & 18071.83 & 28.38 & 143.37 \\
			\hline
		\end{tabular}%
	\end{center}
\end{table}

\subsection{Simulation 2}\label{simu2.sec}

\begin{table}[!h]
	\caption{Empirical sizes (in $\%$) of the HSICp, HSICg, and NEW tests in  Simulation 2.}\label{size2.tab}	
	\begin{center}
		\scalebox{0.85}{
			\begin{tabular}{cccccccccccc}
				\hline
				\multicolumn{3}{c}{}      & \multicolumn{3}{c}{$\rho=0.1$} & \multicolumn{3}{c}{$\rho=0.5$} & \multicolumn{3}{c}{$\rho=0.9$} \\
				\hline
				Model		&$n$     & $p$     & HSICp    & HSICg  & NEW    & HSICp    & HSICg  & NEW    & HSICp    & HSICg  & NEW\\
				\hline
				\multirow{9}[0]{*}{1} & \multirow{3}[0]{*}{30} & 50    & 4.80  & 0.00  & 6.10  & 5.35  & 0.00  & 6.57  & 5.11  & 1.08  & 5.56 \\
				&       & 100   & 4.47  & 0.00  & 5.85  & 4.77  & 0.00  & 6.02  & 4.93  & 0.00  & 5.53 \\
				&       & 200   & 4.98  & 0.00  & 6.14  & 5.28  & 0.00  & 6.35  & 4.94  & 0.00  & 5.65 \\
				\cline{2-12}
				& \multirow{3}[0]{*}{50} & 50    & 4.88  & 0.00  & 5.64  & 5.24  & 0.00  & 5.76  & 5.14  & 2.19  & 5.33 \\
				&       & 100   & 4.74  & 0.00  & 5.56  & 4.59  & 0.00  & 5.38  & 4.91  & 0.04  & 5.37 \\
				&       & 200   & 4.83  & 0.00  & 5.82  & 5.06  & 0.00  & 5.89  & 5.19  & 0.00  & 5.56 \\
				\cline{2-12}
				& \multirow{3}[0]{*}{100} & 50    & 5.02  & 0.00  & 5.42  & 4.73  & 0.00  & 5.11  & 5.03  & 3.64  & 5.15 \\
				&       & 100   & 4.97  & 0.00  & 5.45  & 4.93  & 0.00  & 5.49  & 4.62  & 0.58  & 4.80 \\
				&       & 200   & 5.33  & 0.00  & 5.73  & 5.07  & 0.00  & 5.65  & 4.95  & 0.00  & 5.24 \\
				\hline
				\multirow{9}[0]{*}{2} & \multirow{3}[0]{*}{30} & 50    & 5.11  & 0.00  & 7.12  & 5.34  & 0.00  & 6.78  & 5.19  & 1.44  & 5.89 \\
				&       & 100   & 4.78  & 0.00  & 6.49  & 5.13  & 0.00  & 6.86  & 5.11  & 0.02  & 5.90 \\
				&       & 200   & 4.70  & 0.00  & 6.48  & 5.00  & 0.00  & 6.67  & 4.83  & 0.00  & 5.87 \\
				\cline{2-12}
				& \multirow{3}[0]{*}{50} & 50    & 4.62  & 0.00  & 5.80  & 4.77  & 0.00  & 5.79  & 4.78  & 2.21  & 5.30 \\
				&       & 100   & 4.68  & 0.00  & 5.82  & 5.19  & 0.00  & 6.09  & 4.70  & 0.09  & 5.27 \\
				&       & 200   & 5.03  & 0.00  & 6.02  & 4.81  & 0.00  & 5.75  & 5.20  & 0.00  & 5.81 \\
				\cline{2-12}
				& \multirow{3}[0]{*}{100} & 50    & 4.83  & 0.00  & 5.20  & 4.81  & 0.00  & 5.31  & 5.25  & 3.65  & 5.35 \\
				&       & 100   & 5.03  & 0.00  & 5.41  & 5.27  & 0.00  & 5.83  & 4.96  & 0.88  & 5.13 \\
				&       & 200   & 5.19  & 0.00  & 5.71  & 4.71  & 0.00  & 5.30  & 5.01  & 0.00  & 5.37 \\
				\hline
				\multirow{9}[0]{*}{3} & \multirow{3}[0]{*}{30} & 50    & 5.04  & 0.00  & 6.83  & 4.42  & 0.00  & 5.84  & 5.10  & 1.55  & 5.93 \\
				&       & 100   & 5.17  & 0.00  & 6.21  & 5.08  & 0.00  & 6.38  & 5.31  & 0.04  & 6.04 \\
				&       & 200   & 5.31  & 0.00  & 6.21  & 5.11  & 0.00  & 6.37  & 4.76  & 0.00  & 5.60 \\
				\cline{2-12}
				& \multirow{3}[0]{*}{50} & 50    & 4.69  & 0.00  & 5.48  & 5.25  & 0.00  & 6.18  & 4.83  & 2.40  & 5.20 \\
				&       & 100   & 5.01  & 0.00  & 5.47  & 4.94  & 0.00  & 5.72  & 4.94  & 0.17  & 5.42 \\
				&       & 200   & 5.10  & 0.00  & 5.23  & 4.95  & 0.00  & 5.42  & 4.80  & 0.00  & 5.36 \\
				\cline{2-12}
				& \multirow{3}[0]{*}{100} & 50    & 5.22  & 0.00  & 5.39  & 4.72  & 0.00  & 5.07  & 4.75  & 3.90  & 4.86 \\
				&       & 100   & 5.49  & 0.00  & 5.71  & 4.82  & 0.00  & 5.05  & 5.01  & 1.06  & 5.17 \\
				&       & 200   & 4.62  & 0.00  & 4.45  & 5.09  & 0.00  & 5.34  & 4.78  & 0.00  & 5.11 \\
				\hline
				\multicolumn{3}{c}{ARE} & 4.04  & 100.00 & 16.92 & 4.07  & 100.00 & 17.01 & 3.09  & 81.53 & 9.22 \\
				\hline
		\end{tabular}}
	\end{center}
\end{table}

In this simulation study, we  compare  the HSICp, HSICg, and NEW  tests  for high-dimensional data in terms of size control and power.  In each simulation run, the sample  $(\bx_i, \by_i), i=1,\ldots,n$ are generated as follows.  For size control comparison, we set $\bx_i=\bGamma \bz_{1i},i=1,\ldots,n$ and $\by_i=\bGamma \bz_{2i},i=1,\ldots,n$ where $\bGamma\bGamma^\T =\bSigma/\tr(\bSigma^2)$, and $\bSigma=(\sigma_{st})_{s,t=1}^p,\sigma_{st}=\rho^{|s-t|}$ with $\rho$ controlling the data correlation. The i.i.d. random variables $\bz_{k i}=(z_{k i 1},\ldots, z_{k i p})^\T, i=1,\ldots,n;\;k=1,2$ are generated from the following three models:
\begin{description}
	\item[Model 1.] $z_{k ij},j=1,\ldots,p\iidsim N(0,1)$.
	\item[Model 2.] $z_{k ij}=w_{k ij}/\sqrt{2}, j=1,\ldots,p$, where $w_{k ij}\iidsim t_4$, the $t$-distribution with 4 degrees of freedom.
	\item[Model 3.] $z_{k ij}=(w_{k ij}-1)/\sqrt{2},j=1,\ldots,p$, where $w_{k ij}\iidsim \chi_1^2$.
\end{description}
For power comparison, we set  $\bx_i=\bGamma \bz_{1i},i=1,\ldots,n$  where $\bz_{1i}=(z_{1i1},\ldots, z_{1ip})^\T, i=1,\ldots,n$ with $z_{1ij}, j=1,\ldots, p; i=1,\ldots,n$ generated from  Model 1 and set $\by_i=(y_{i1},\ldots,y_{ip})^\T$ with $y_{ij}=\delta(x_{ij}+x_{ij}^2)+z_{2ij},j=1,\ldots,p;i=1,\ldots,n$,  with $x_{ij}$ being the $j$th entry of $\bx_i$ for $j=1,\ldots,p; i=1,\ldots, n$ and  $z_{2ij}, j=1,\ldots, p; i=1,\ldots,n$ generated from Model 2. We set  $\delta\in\{0.6,0.8,1.0,1.2\}$. Note that as $\delta$ increases, the dependence between $\bx_i$ and $\by_i$ are getting stronger so that  the power of a test should increase. We specify $\rho=0.1,0.5$, and $0.9$ to represent less, moderately, and highly correlated data cases. To generate the high-dimensional data under  ``large $p$, small $n$" settings, we choose $n=30,50,100$ and   $p=50,100,200$, respectively.  Here and throughout, to measure the overall performance of a test in maintaining the nominal size $\alpha=5\%$,  we employ the  average relative error (ARE) criterion  of \cite{zhang2011two}. The ARE value of a test is calculated as ARE $=100M^{-1}\sum_{j=1}^M|\halpha_j-\alpha|/\alpha$, where $\halpha_j,j=1,\ldots,M$, denote the empirical sizes under $M$ simulation settings. A smaller ARE value of a test indicates a better performance of that test in terms of size control.

\begin{table}[!h]
	\caption{Empirical powers (in $\%$) of  the HSICp, HSICg, and NEW tests in  Simulation 2.}\label{power2.tab}	
	\begin{center}
			\setlength{\tabcolsep}{3pt}	
			\scalebox{0.85}{
			\begin{tabular*}{\textwidth}{@{\extracolsep{\fill}}cccccccccccc}
				\hline
				\multicolumn{3}{c}{}      & \multicolumn{3}{c}{$\rho=0.1$} & \multicolumn{3}{c}{$\rho=0.5$} & \multicolumn{3}{c}{$\rho=0.9$} \\
				\hline
				$n$     & $p$   &$\delta$  & HSICp    & HSICg  & NEW    & HSICp    & HSICg  & NEW    & HSICp    & HSICg  & NEW\\
				\hline
		\multirow{12}[0]{*}{30} & \multirow{4}[0]{*}{50} & 0.6   & 15.29 & 0.00  & 18.22 & 15.54 & 0.00  & 18.21 & 16.11 & 0.00  & 18.51 \\
		&       & 0.8   & 27.17 & 0.00  & 31.47 & 29.05 & 0.00  & 33.05 & 29.42 & 0.01  & 32.63 \\
		&       & 1.0   & 45.64 & 0.00  & 50.97 & 47.45 & 0.00  & 52.65 & 49.11 & 0.01  & 52.60 \\
		&       & 1.2   & 66.95 & 0.00  & 71.14 & 68.81 & 0.00  & 73.15 & 70.86 & 0.09  & 73.93 \\
		\cline{2-12}
		& \multirow{4}[0]{*}{100} & 0.6   & 10.42 & 0.00  & 13.02 & 10.27 & 0.00  & 12.96 & 11.21 & 0.00  & 13.49 \\
		&       & 0.8   & 16.17 & 0.00  & 19.39 & 16.42 & 0.00  & 20.18 & 17.48 & 0.00  & 20.50 \\
		&       & 1.0   & 26.27 & 0.00  & 30.63 & 26.53 & 0.00  & 31.17 & 28.84 & 0.00  & 32.58 \\
		&       & 1.2   & 39.93 & 0.00  & 45.16 & 42.59 & 0.00  & 47.70 & 45.22 & 0.00  & 49.55 \\
			\cline{2-12}
		& \multirow{4}[0]{*}{200} & 0.6   & 8.06  & 0.00  & 9.89  & 8.43  & 0.00  & 10.31 & 7.76  & 0.00  & 9.74 \\
		&       & 0.8   & 10.58 & 0.00  & 13.03 & 11.03 & 0.00  & 13.79 & 10.97 & 0.00  & 13.65 \\
		&       & 1.0   & 15.75 & 0.00  & 19.08 & 15.49 & 0.00  & 18.66 & 17.22 & 0.00  & 20.18 \\
		&       & 1.2   & 23.53 & 0.00  & 27.44 & 22.99 & 0.00  & 27.53 & 25.06 & 0.00  & 29.26 \\
		\hline
		\multirow{12}[0]{*}{50} & \multirow{4}[0]{*}{50} & 0.6   & 26.75 & 0.00  & 29.53 & 26.48 & 0.00  & 28.82 & 27.93 & 0.04  & 29.77 \\
		&       & 0.8   & 52.32 & 0.00  & 55.93 & 54.71 & 0.00  & 57.84 & 55.54 & 0.09  & 57.62 \\
		&       & 1.0   & 81.09 & 0.00  & 83.12 & 82.04 & 0.00  & 84.24 & 81.72 & 0.30  & 83.23 \\
		&       & 1.2   & 95.68 & 0.00  & 96.41 & 96.74 & 0.00  & 97.18 & 95.96 & 0.97  & 96.36 \\
			\cline{2-12}
		& \multirow{4}[0]{*}{100} & 0.6   & 15.48 & 0.00  & 17.57 & 15.36 & 0.00  & 17.29 & 16.61 & 0.00  & 18.47 \\
		&       & 0.8   & 28.26 & 0.00  & 31.42 & 29.38 & 0.00  & 32.48 & 32.07 & 0.00  & 34.43 \\
		&       & 1.0   & 49.88 & 0.00  & 53.68 & 51.26 & 0.00  & 55.24 & 54.97 & 0.00  & 57.73 \\
		&       & 1.2   & 72.03 & 0.00  & 74.95 & 76.19 & 0.00  & 78.89 & 78.81 & 0.01  & 80.73 \\
			\cline{2-12}
		& \multirow{4}[0]{*}{200} & 0.6   & 9.84  & 0.00  & 11.49 & 10.38 & 0.00  & 12.05 & 10.23 & 0.00  & 11.74 \\
		&       & 0.8   & 16.20 & 0.00  & 18.30 & 16.66 & 0.00  & 18.98 & 17.56 & 0.00  & 19.90 \\
		&       & 1.0   & 26.55 & 0.00  & 29.66 & 27.42 & 0.00  & 30.54 & 30.16 & 0.00  & 33.10 \\
		&       & 1.2   & 42.78 & 0.00  & 46.14 & 44.97 & 0.00  & 48.83 & 47.36 & 0.00  & 50.70 \\
		\hline
		\multirow{12}[0]{*}{100} & \multirow{4}[0]{*}{50} & 0.6   & 62.82 & 0.00  & 64.73 & 64.45 & 0.00  & 66.01 & 63.28 & 0.25  & 64.19 \\
		&       & 0.8   & 95.04 & 0.00  & 95.65 & 95.64 & 0.00  & 95.95 & 94.30 & 1.05  & 94.55 \\
		&       & 1.0   & 99.93 & 0.00  & 99.95 & 99.99 & 0.07  & 99.99 & 99.78 & 3.19  & 99.81 \\
		&       & 1.2   & 100.00 & 0.00  & 100.00 & 100.00 & 0.15  & 100.00 & 100.00 & 15.96 & 100.00 \\
			\cline{2-12}
		& \multirow{4}[0]{*}{100} & 0.6   & 33.23 & 0.00  & 35.29 & 34.88 & 0.00  & 36.72 & 36.05 & 0.00  & 37.60 \\
		&       & 0.8   & 67.56 & 0.00  & 69.55 & 69.28 & 0.00  & 71.21 & 72.47 & 0.03  & 73.59 \\
		&       & 1.0   & 93.29 & 0.00  & 93.94 & 94.91 & 0.00  & 95.55 & 95.23 & 0.14  & 95.59 \\
		&       & 1.2   & 99.64 & 0.00  & 99.68 & 99.75 & 0.00  & 99.78 & 99.81 & 0.22  & 99.85 \\
			\cline{2-12}
		& \multirow{4}[0]{*}{200} & 0.6   & 18.29 & 0.00  & 19.64 & 18.07 & 0.00  & 19.33 & 20.01 & 0.00  & 21.32 \\
		&       & 0.8   & 35.98 & 0.00  & 38.15 & 37.69 & 0.00  & 40.10 & 39.87 & 0.00  & 41.77 \\
		&       & 1.0   & 63.61 & 0.00  & 65.94 & 65.66 & 0.00  & 67.87 & 69.62 & 0.00  & 71.48 \\
		&       & 1.2   & 87.67 & 0.00  & 89.02 & 89.35 & 0.00  & 90.67 & 92.13 & 0.01  & 92.92 \\
		\hline
	\end{tabular*}}
\end{center}
\end{table}

Table~\ref{size2.tab} presents the empirical sizes of the HSICp, HSICg, and NEW tests  under various settings, with the last row denoting  their ARE values associated with the  three values of $\rho$. In terms of size control, both the HSICp and NEW tests perform well regardless of whether the data are less correlated ($\rho=0.1$), moderately correlated ($\rho=0.5$), or highly correlated $(\rho=0.9)$ since  their empirical sizes are generally around 5\% and their ARE values are generally below $20$. They perform much better than the HSICg test which does not work at all with its empirical sizes being $0$  when $\rho=0.1$ and  0.5, and is very conservative with its empirical sizes being much smaller than $5\%$ when $\rho=0.9$. These conclusions are consistent with those observed from Simulation 1 when the data dimension is small or moderate.

The empirical powers of the HSICp, HSICg, and NEW tests under various configurations  are presented in Table~\ref{power2.tab}. It is seen that  for fixed $n$ and $p$, as the value of $\delta$ increases, the empirical powers of the three tests increase, and under each setting, the empirical powers of the NEW test are generally  larger than those of the HSICp test,  while  the HSICg test is totally no power. The power magnitude order of the three tests is obviously affected by  that of  their empirical sizes, as seen  from Table~\ref{size2.tab}.

\subsection{Simulation 3}\label{simu3.sec}

In this simulation study, we compare  the  NEW test against a few representative existing tests for the two-sample independence testing problem for functional data.  These existing tests include
the  Pearson correlation based test,  dynamical correlation  based test (\citealt{dubin2005dynamical}), 
FPCA-based distance covariance test (\citealt{kosorok2009discussion}), and global temporal correlation  based test (\citealt{zhou2018local}),  denoted as Pearson, dnm, FPCA, and gtemp, respectively. The dnm, FPCA, and gtemp tests  are permutation based  and are implemented in \cite{miao2022wavelet}.

 In each simulation run, the functional observations  $\{x_i(t),y_i(t)\},t\in [0,1],i=1,\ldots,n$ are generated as follows. We take $x_i(t)=\sum_{j=1}^{50}z_{1ij}\sqrt{2}\cos(j\pi t)$ and
 $y_i(t)=\sum_{j=1}^{50}z_{2ij}\sqrt{2}\cos(j\pi t)$, where $z_{1ij},j=1,\ldots,50;i=1,\ldots,n$ are i.i.d. random variables generated in the same way as described in Simulation 2,
 and for some give function $f(\cdot)$, we set $z_{2ij}=f(z_{1ij})$ for $j=1,\ldots,m$ and $z_{2ij}\iidsim N(0,1)$ for $j=m+1,\ldots,50$ where $m$ is an integer used to control  the dependency level of the functional observations such that when $m=0$, the null hypothesis holds and otherwise the dependency level increases with $m$.  In numerical implementation, these functional observations will be evaluated at a grid of equal-spaced time points: $t_r=(r-1)/(k-1),r=1,\ldots,k$.  For power consideration,
 we consider four functions: $f(u)=u^3, f(u)=u^2, f(u)=u\sin(u)$, and $f(u)=u\cos(u)$, with  the first one being monotone and the  last three  are not; see Figure~\ref{funs.fig} for some details. It is often more difficult to
 detect the dependency when $f(u)=u\sin(u)$ and $f(u)=u\cos(u)$ than when $f(u)=u^3$ and $f(u)=u^2$. Thus we set $m=3,5,10$, and $15$ when $f(u)=u^3$ and $f(u)=u^2$, and set  $m=15,20,25$,
 and $30$ when $f(u)=u\sin(u)$ and $f(u)=u\cos(u)$. 

\begin{figure}[!h]
	\centering\includegraphics[width=0.8\linewidth]{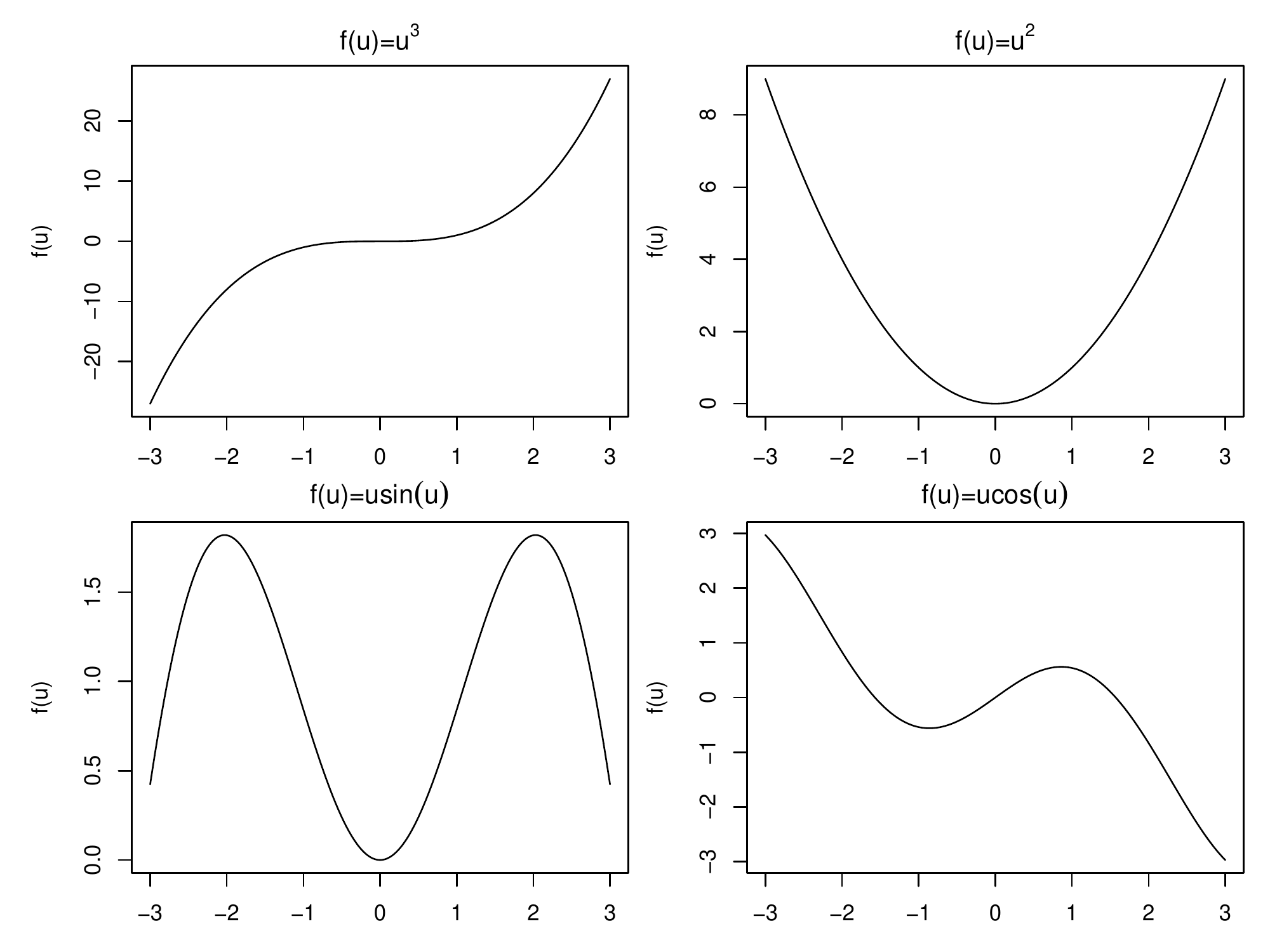}
	\caption{\em Four functions considered in Simulation 3. }\label{funs.fig}
\end{figure}

Table~\ref{size3.tab} displays the empirical sizes of the five considered tests with the last row being the associated ARE values. It is seen that all the five tests have good level accuracy with their empirical sizes generally around 5\%. Admittedly, in terms of size control, the NEW test performs slightly worse than the other four tests but it also computes much faster than them.

\begin{table}[!h]
	\caption{Empirical sizes (in $\%$) of the Pearson, dnm, FPCA, gtemp, and NEW tests in  Simulation 3.}\label{size3.tab}	
	\begin{center}
		\scalebox{0.85}{
			\begin{tabular*}{\textwidth}{@{\extracolsep{\fill}}cccccccc}
				\hline
				Model  & $n$     & $k$     & Pearson & dnm   & FPCA  & gtemp & NEW \\
				\hline
				\multirow{9}[0]{*}{1} & \multirow{3}[0]{*}{30} & 201   & 4.89  & 5.12  & 4.68  & 5.02  & 6.88 \\
				&       & 501   & 5.13  & 5.08  & 4.42  & 5.19  & 6.09 \\
				&       & 1001  & 5.17  & 5.39  & 4.70  & 4.96  & 6.13 \\
				\cline{2-8}
				& \multirow{3}[0]{*}{50} & 201   & 4.83  & 4.94  & 4.36  & 5.21  & 5.69 \\
				&       & 501   & 4.92  & 4.86  & 4.56  & 5.08  & 5.98 \\
				&       & 1001  & 5.02  & 4.90  & 4.55  & 4.89  & 5.92 \\
				\cline{2-8}
				& \multirow{3}[0]{*}{100} & 201   & 4.84  & 5.07  & 4.34  & 5.01  & 5.45 \\
				&       & 501   & 4.84  & 4.91  & 4.59  & 5.14  & 5.72 \\
				&       & 1001  & 4.79  & 5.04  & 4.78  & 5.13  & 5.48 \\
				\hline
				\multirow{9}[0]{*}{2} & \multirow{3}[0]{*}{30} & 201   & 4.55  & 4.73  & 4.79  & 4.77  & 6.29 \\
				&       & 501   & 5.06  & 4.65  & 4.26  & 5.19  & 6.50 \\
				&       & 1001  & 5.11  & 4.89  & 4.77  & 4.85  & 6.21 \\
				\cline{2-8}
				& \multirow{3}[0]{*}{50} & 201   & 4.79  & 4.86  & 4.92  & 4.88  & 6.25 \\
				&       & 501   & 4.93  & 4.73  & 4.35  & 4.83  & 5.62 \\
				&       & 1001  & 5.52  & 5.35  & 4.37  & 4.95  & 6.31 \\
				\cline{2-8}
				& \multirow{3}[0]{*}{100} & 201   & 5.26  & 5.33  & 4.41  & 5.37  & 5.08 \\
				&       & 501   & 5.13  & 5.00  & 4.65  & 5.13  & 5.31 \\
				&       & 1001  & 4.84  & 4.99  & 4.47  & 4.57  & 5.84 \\
				\hline
				\multirow{9}[0]{*}{3} & \multirow{3}[0]{*}{30} & 201   & 5.27  & 5.17  & 4.62  & 4.79  & 6.17 \\
				&       & 501   & 5.36  & 5.29  & 4.09  & 4.93  & 6.23 \\
				&       & 1001  & 5.07  & 4.98  & 4.40  & 5.16  & 6.28 \\
				\cline{2-8}
				& \multirow{3}[0]{*}{50} & 201   & 5.11  & 5.04  & 4.57  & 5.03  & 5.95 \\
				&       & 501   & 5.13  & 5.03  & 4.96  & 5.25  & 6.06 \\
				&       & 1001  & 4.78  & 4.83  & 4.62  & 4.96  & 5.39 \\	
				\cline{2-8}
				& \multirow{3}[0]{*}{100} & 201   & 4.96  & 5.20  & 4.52  & 4.94  & 5.33 \\
				&       & 501   & 5.55  & 5.65  & 4.46  & 5.15  & 5.51 \\
				&       & 1001  & 4.87  & 4.91  & 4.34  & 5.01  & 5.41 \\
				\hline
				\multicolumn{3}{c}{ARE} & 3.75  & 3.39  & 9.22  & 2.78  & 17.84 \\
				\hline
		\end{tabular*}}
	\end{center}
\end{table}

For  space saving,  Table~\ref{power3.tab} only displays the empirical powers (in \%) of the five considered tests under Model 1 with $k=201$ since the conclusions drawn for other values of $k$ are similar. These empirical powers are quite revealing in several ways. First of all, for monotone relationship (i.e., when $f(u)=u^3$),  the NEW test is just slightly less powerful than the other four tests.  The slight  power inferiority of the NEW test is possibly due to the fact that the NEW test is an HSIC-based test  since,  as pointed out by  \cite{shen2019distance},   to detect a monotone relationship,  the  HSIC based tests  may be  slightly inferior to the distance covariance based tests. Second,  for non-monotone relationship (i.e., when $f(u)=u^2, u\sin(u)$ and $u\cos(u)$), the NEW test is generally more powerful than the other four tests. The performances of the other four tests are quite different for different non-monotone relationships. For example, the Pearson test  has almost no powers when $f(u)=u^2$ and $u\sin(u)$, the gtemp test has almost no powers when $f(u)=u\sin(u)$ and $u\cos(u)$, while the dnm and FPCA tests have very low powers when $f(u)=u\sin(u)$.
\begin{table}[!h]
	\caption{Empirical powers (in $\%$) of the Pearson, dnm, FPCA, gtemp, and NEW tests in  Simulation 3 under Model 1 with $k=201$.}\label{power3.tab}	
	\begin{center}
		\scalebox{0.85}{
			\begin{tabular*}{\textwidth}{@{\extracolsep{\fill}}cccccccccccc}
				\hline
				&       & \multicolumn{5}{c}{$f(u)=u^3$}         & \multicolumn{5}{c}{$f(u)=u^2$} \\
				\hline
				$n$     & $m$ & Pearson & dnm   & FPCA  & gtemp & NEW & Pearson & dnm   & FPCA  & gtemp & NEW \\
				\hline
				\multirow{4}[0]{*}{30} & 3     & 91.32 & 95.98 & 83.51 & 75.89 & 54.25 & 4.90  & 8.25  & 8.03  & 5.98  & 9.68 \\
				& 5     & 99.90 & 99.97 & 97.07 & 99.04 & 79.72 & 5.28  & 9.98  & 10.96 & 6.47  & 12.83 \\
				& 10    & 100 & 100 & 99.98 & 100 & 97.62 & 5.01  & 14.19 & 18.25 & 8.45  & 23.33 \\
				& 15    & 100 & 100 & 99.99 & 100 & 99.65 & 4.49  & 17.41 & 26.60 & 9.22  & 37.48 \\
				\hline
				\multirow{4}[0]{*}{50} & 3     & 99.19 & 99.73 & 99.28 & 96.00 & 79.70 & 4.80  & 7.59  & 8.90  & 5.83  & 11.02 \\
				& 5     & 100 & 100 & 100 & 99.99 & 97.47 & 5.18  & 10.17 & 11.89 & 7.03  & 15.70 \\
				& 10    & 100 & 100 & 100 & 100 & 99.99 & 4.89  & 14.61 & 20.87 & 8.34  & 31.80 \\
				& 15    & 100 & 100 & 100 & 100 & 100 & 4.54  & 16.96 & 32.32 & 10.16 & 53.21 \\
				\hline
				\multirow{4}[0]{*}{100} & 3     & 100 & 100 & 100 & 99.99 & 99.16 & 4.83  & 7.77  & 9.54  & 5.63  & 14.37 \\
				& 5     & 100 & 100 & 100 & 100 & 100 & 4.73  & 9.38  & 14.02 & 6.60  & 23.82 \\
				& 10    & 100 & 100 & 100 & 100 & 100 & 4.84  & 13.62 & 27.72 & 8.38  & 57.57 \\
				& 15    & 100 & 100 & 100 & 100 & 100 & 4.73  & 17.70 & 42.68 & 10.40 & 85.17 \\
				\hline
				&       & \multicolumn{5}{c}{$f(u)=u\sin(u)$}      & \multicolumn{5}{c}{$f(u)=u\cos(u)$} \\
				\hline
				$n$& $m$ & Pearson & dnm   & FPCA  & gtemp & NEW & Pearson & dnm   & FPCA  & gtemp & NEW \\
				\hline
				\multirow{4}[0]{*}{30}    & 15    & 4.59  & 6.42  & 5.53  & 4.91  & 9.63  & 5.44  & 8.53  & 6.90  & 5.14  & 9.46 \\
				& 20    & 5.24  & 6.95  & 6.14  & 5.64  & 11.35 & 6.22  & 11.12 & 7.84  & 6.12  & 11.15 \\
				& 25    & 4.82  & 7.77  & 6.95  & 5.93  & 14.07 & 6.67  & 14.16 & 9.85  & 6.53  & 14.42 \\
				& 30    & 4.71  & 9.42  & 8.27  & 6.13  & 17.02 & 8.81  & 19.15 & 11.86 & 7.51  & 18.55 \\
				\hline
				\multirow{4}[0]{*}{50}    & 15    & 5.13  & 6.82  & 5.84  & 5.35  & 10.08 & 5.71  & 9.28  & 6.73  & 5.74  & 9.74 \\
				& 20    & 5.00  & 7.65  & 6.52  & 5.67  & 13.44 & 7.89  & 13.35 & 8.08  & 6.26  & 12.75 \\
				& 25    & 5.07  & 8.55  & 7.79  & 5.36  & 17.92 & 9.12  & 17.51 & 10.04 & 6.66  & 16.95 \\
				& 30    & 5.09  & 9.44  & 8.31  & 6.29  & 24.27 & 11.49 & 23.39 & 11.12 & 7.20  & 23.15 \\
				\hline
				\multirow{4}[0]{*}{100}   & 15    & 5.03  & 6.62  & 6.47  & 5.24  & 15.33 & 7.19  & 10.89 & 6.82  & 5.49  & 12.11 \\
				& 20    & 5.38  & 7.46  & 7.10  & 4.96  & 21.69 & 9.60  & 16.10 & 7.77  & 6.13  & 16.71 \\
				& 25    & 4.84  & 8.93  & 8.38  & 5.83  & 32.36 & 13.27 & 23.21 & 9.55  & 6.63  & 25.10 \\
				& 30    & 4.94  & 9.56  & 9.51  & 5.76  & 44.74 & 18.38 & 33.68 & 11.40 & 7.31  & 36.51 \\
				\hline
		\end{tabular*}}
	\end{center}
\end{table}

 From the above three simulation studies, in terms of level accuracy,  power,  and computational costs, the NEW test outperforms the other  competitors generally  and hence it should be recommended in real data analysis.

\section{Applications to functional and high-dimensional data}\label{real_data.sec}

In this section, we present the applications of the NEW test, together with several existing competitors mentioned in the previous section, to a functional data set and a high-dimensional data set. For the NEW test, we continue to use the Gaussian RBF kernel and choose the kernel width as described in the previous section.


\subsection{Canadian weather data}

\begin{figure}[!h]
	\centering\includegraphics[width=0.8\linewidth]{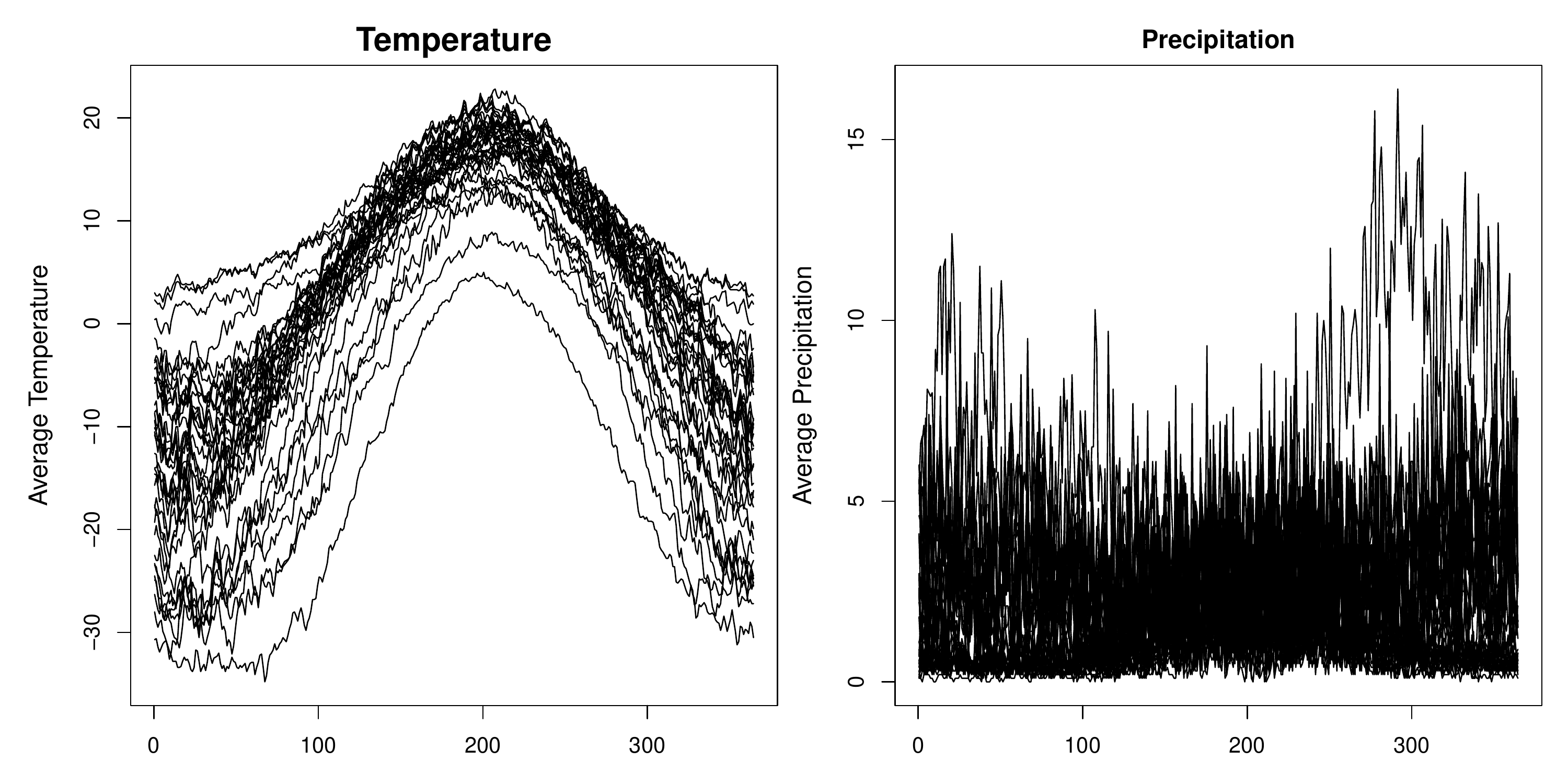}
	\caption{\em Raw temperature and precipitation curves for $35$ Canadian weather stations.}\label{Canadianplot}
\end{figure}

In this subsection, we illustrate the applications of the Pearson, dnm, FPCA, gtemp and NEW  tests to functional data using
the Canadian weather data set, briefly introduced in Section~\ref{intro.sec}. For each of the 35 weather stations, over a year period,  the variable ``Temperature''  records the average daily temperature  and the variable ``Precipitation'' records the average daily rainfall  rounded to $0.1$ mm.  The raw temperature and precipitation curves for the $35$ weather stations are presented in Figure~\ref{Canadianplot}. It is expected that there is some dependency between   the average daily temperature and the average daily precipitation since they  were recorded from the same $35$  Canadian weather stations. Of interest is to check how strong this dependency is.

 To this end, we apply the Pearson, dnm, FPCA, gtemp, and NEW test to this Canadian weather data set  to check whether the  temperature curves and the   precipitation curves are independent. The $p$-values of the five  tests are shown in Table~\ref{canawea.tab}. It is seen that  the $p$-values of all the five  tests are quite small and much smaller than 1\%,  suggesting  that there is some strong evidence to reject the null hypothesis, i.e.,  there is strong dependency  between the  temperature curves and the   precipitation curves for  the $35$ Canadian weather stations, as expected.
\begin{table}[!h]
	\caption{$p$-values of the Pearson, dnm, FPCA, gtemp, and NEW tests for testing the independence between the underlying temperature curves and the underlying  precipitation curves.\label{canawea.tab}}
	\smallskip
	\centering\small
	\scalebox{0.85}{
		\begin{tabular*}{\textwidth}{@{\extracolsep{\fill}}ccccc}
			\hline
			Pearson &dnm &FPCA &gtemp  &NEW\\
			\hline
			$2.73\times 10 ^{-3}$ &0 &$5.00\times 10^{-3}$ &0  &$8.29\times 10^{-7}$\\
			\hline
	\end{tabular*}}
\end{table}
\subsection{Colon data}

In this subsection, we illustrate the applications of the  HSICp,  HSICg, and NEW tests to  high-dimensional data using the well-known colon data set which   contains $62$  tissues, each having $2000$ gene expression levels, and can be downloaded from \url{http://microarray.princeton.edu/oncology/affydata/index.html}. In order to construct a two-sample test for independence for high-dimensional data, we choose  the first 31 tissues to form the first group, and the remaining  31 tissues to form the second group. Thus, the two groups should be independent since these tissues are independent. Table~\ref{colon.tab1} displays the $p$-values of the three considered tests, which are all larger than 50\%, showing that  there is no  evidence at all  to reject the null hypothesis, as expected.

\begin{table}[!h]
	\caption{$p$-values of the HSICp, HSICg, and NEW tests for testing the independence between the two groups of the colon data.\label{colon.tab1}}
	\smallskip
	\centering
	\scalebox{0.85}{
	\begin{tabular*}{\textwidth}{@{\extracolsep{\fill}}ccc}
		\hline
		HSICp & HSICg &NEW\\
		\hline
		0.956 &0.619 &0.560\\
		\hline
	\end{tabular*}}
\end{table}

To further demonstrate the level accuracy of the  NEW test against the HSICp and HSICg tests, a small scale simulation study based on this colon data set is conducted to simulate the empirical sizes
of the three tests,  obtained from $10,000$ simulation runs. In each run,  the 62 tissues are randomly  split  into two groups of equal-size. The empirical size of a test is calculated  as the proportion of times when  the $p$-value of the test is smaller than the nominal size  $\alpha=5\%$ or $10\%$ based on the $10,000$ runs. The empirical sizes of the three tests are displayed in Table~\ref{colon.tab2}. It is seen that both the  NEW and HSICp tests   have  good  level accuracy but the  HSICg test is rather conservative. This is consistent with the conclusions drawn from the simulation results presented in Tables~\ref{size1.tab} and~\ref{size2.tab}.

\begin{table}[!h]
	\caption{Empirical sizes (in \%) of the HSICp, HSICg, and NEW  tests obtained from the small scale simulation study.\label{colon.tab2}}
	\smallskip
	\centering
	\scalebox{0.85}{
	\begin{tabular*}{\textwidth}{@{\extracolsep{\fill}}cccc}
		\hline
	$\alpha$	&HSICp & HSICg &NEW\\
		\hline
	5\%	&4.96 &2.25 &6.80\\
	10\% &9.76 &5.71 &11.60\\
		\hline
	\end{tabular*}}
\end{table}

\section{Concluding remarks}\label{con.sec}

In the literature, several tests  have been proposed  for  two-sample independence test in separable metric spaces based on the Hilbert--Schmidt Independence Criterion (HSIC).
In this paper, we propose and study a new HSIC based independence  test in separable metric spaces with applications to functional and high-dimensional data. Under some regularity conditions and the null hypothesis, it is shown that the proposed test statistic asymptotically  has  a chi-squared-type mixture limit.  To conduct the proposed test, we employ the three-cumulant matched chi-squared-approximation of \cite{zhang2005approximate} to approximate the distribution of the chi-squared-mixture  with the approximation parameters consistently estimated from the data. Simulation studies and real data applications demonstrate that in terms of size control, power, and computational cost,  the proposed  test  outperforms several existing tests for multivariate, high-dimensional, and functional data. Nevertheless, Tables~\ref{size2.tab} and~\ref{size3.tab} also indicate that  the proposed  test is still somewhat liberal when the sample sizes are  small. Methods for further improving the level accuracy  of the proposed test are interesting and warranted.

\section*{Appendix\label{sec:Appendix}}

\setcounter{equation}{0} \global\long\def\theequation{A.\arabic{equation}}
\begin{proof}[Proof of Lemma~\ref{lem1.sec2}]
	With the expression of $\tK(x_i,x_j)$ given in (\ref{tK.sec2}) and $\tK^*(x_i,x_j)$ given in (\ref{tKstar.sec2}), for any fixed $i,j\in\{1,\ldots,n\}$, we have
	\[
	\begin{split}
		|\tK^*(x_i,x_j)-\tK(x_i,x_j)|&\leq|n^{-1}(n-1)^{-1}\sum_{u\neq v}K(x_u,x_v)-\E_{z,z'}[K(z,z')]|\\
		&+|(n-1)^{-1}\sum_{v\neq i}K(x_i,x_v)- \E_{z'}[K(x_i,z')]|+|(n-1)^{-1}\sum_{u\neq j}K(x_u,x_j)-\E_{z}[K(z,x_j)]|.
	\end{split}
	\]
	It follows that as $n\to\infty$, we have
	\[
	\begin{split}
		\E\{n[n^{-1}(n-1)^{-1}\sum_{u\neq v}K(x_u,x_v)-\E_{z,z'}[K(z,z')]]\}&=o(1),\;\mbox{ and }\\
		\Var\{n[n^{-1}(n-1)^{-1}\sum_{u\neq v}K(x_u,x_v)-\E_{z,z'}[K(z,z')]]\}&\leq\E[K(x,x')]^2\leq 16B_K^2.
	\end{split}
	\]
	Therefore, as $n\to\infty$, we have $|n^{-1}(n-1)^{-1}\sum_{u\neq v}K(x_u,x_v)-\E_{z,z'}[K(z,z')]|=\mathcal{O}(n^{-1})$. Similarly, we have $|(n-1)^{-1}\sum_{v\neq i}K(x_i,x_v)- \E_{z'}[K(x_i,z')]|=\mathcal{O}(n^{-1/2})$ and $|(n-1)^{-1}\sum_{u\neq j}K(x_u,x_j)-\E_{z}[K(z,x_j)]|=\mathcal{O}(n^{-1/2})$. Hence, as $n\to \infty$, we have
	\[
	\tK^*(x_i,x_j)=\tK(x_i,x_j)+\mathcal{O}(n^{-1/2})\; \mbox{uniformly for all } x_i,x_j.
	\]
	The lemma is then proved.
\end{proof}	

\begin{proof}[Proof of Theorem~\ref{tTndist2.thm}] Under the null hypothesis, $x$ and $y$ are independent. Then under the condition (\ref{cond.sec2}) and by (\ref{mercer.sec2}), we have
\[
\E[\tK(x,x)\tL(y,y)]=\E[\tK(x,x)]\E[\tL(y,y)]=\left(\sum\limits_{r=1}^{\infty}\lambda_r\right)\left(\sum\limits_{s=1}^{\infty}\rho_s\right)<\infty,
 \]
 and
 \[
 \E[\tK(x,x')\tL(y,y')]^2= \E[\tK(x,x')]^2\E[\tL(y,y')]^2=\left(\sum\limits_{r=1}^{\infty}\lambda_r^2\right)\left(\sum\limits_{s=1}^{\infty}\rho_s^2\right)<\infty.
 \]
 Furthermore, we have
 \[
 \begin{array}{rcl}
 \tK(x,x')\tL(y,y')&=&\sum\limits_{r=1}^{\infty}\sum\limits_{s=1}^{\infty} \lambda_r\rho_s \phi_r(x)\phi_r(x') \psi_s(y)\psi_s(y')\\
 &=&\sum\limits_{r=1}^{\infty}\sum\limits_{s=1}^{\infty} \lambda_r\rho_s [\phi_r(x)\psi_s(y)][\phi_s(x')\psi_s(y')],
 \end{array}
 \]
 where $(x',y')$ is an independent copy of $(x,y)$. Since $x$ and $y$ are independent, we can show that $\phi_r(x)\psi_s(y)$ are orthonormal. In fact, we have
 \[
 \begin{split}
 \E\{[\phi_r(x)\psi_s(y)][\phi_{\alpha}(x)\psi_{\beta}(y)]\}&=\E\{[\phi_r(x)\phi_{\alpha}(x)][\psi_s(y)\psi_{\beta}(y)]\}\\
 &=\E[\phi_r(x)\phi_{\alpha}(x)]\E[\psi_s(y)\psi_{\beta}(y)]=\delta_{r\alpha}\delta_{s\beta},
 \end{split}
 \]
 which takes $1$ if $(r,s)=(\alpha,\beta)$ and $0$ otherwise. If follows that $\lambda_r\rho_s, \; r,s=1,2,\ldots$ are the eigen-elements of $\tK(x,x')\tL(y,y')$, associated with the eigen-functions
 $\phi_r(x)\psi_s(y), r,s=1,2,\dots$.

We can rewrite $\tT_n$ as
\[
\tT_{n}=n^{-1}\sum\limits_{i=1}^n[\tK(x_i,x_i)\tL(y_i,y_i)]+(n-1)\binom{n}{2}^{-1}\sum\limits_{1\le i<j\le n} [\tK(x_i,x_j)\tL(y_i,y_j)].
 \]
 By the law of large numbers, as $n\to\infty$,  we have $n^{-1}\sum\limits_{i=1}^n[\tK(x_i,x_i)\tL(y_i,y_i)]\convas \E[\tK(x,x)\tL(y,y)]=\left(\sum\limits_{r=1}^{\infty} \lambda_r\right)\left(\sum\limits_{s=1}^{\infty} \rho_s\right)$. Further, since \[
 \E_{x,y}[\tK(x,x')\tL(y,y')]=\E_{x',y'}[\tK(x,x')\tL(y,y')]=\E_{x,y,x',y'}[\tK(x,x')\tL(y,y')]=0
  \]
  and $\E[\tK(x,x')\tL(y,y')]^2=\left(\sum\limits_{r=1}^{\infty} \lambda_r^2\right)\left(\sum\limits_{s=1}^{\infty} \rho_s^2\right)<\infty$, by the U-statistics theorem of \citet[p. 194]{serfling}, as $n\to\infty$, we have
$
(n-1)\binom{n}{2}^{-1}\sum\limits_{1\le i<j\le n} [\tK(x_i,x_j)\tL(y_i,y_j)]\convL \sum\limits_{r=1}^{\infty}\sum\limits_{s=1}^{\infty}\lambda_r\rho_s(A_{rs}-1), \;A_{rs}\iidsim \chi_1^2.
$
It follows that $\tT_{n}\convL \tT$ where $\tT=\sum\limits_{r=1}^{\infty} \sum\limits_{s=1}^{\infty}\lambda_r\rho_s+\sum\limits_{r=1}^{\infty}\sum\limits_{s=1}^{\infty}\lambda_r\rho_s(A_{rs}-1)=
\sum\limits_{r=1}^{\infty}\sum\limits_{s=1}^{\infty}\lambda_r\rho_s A_{rs}, \;\; \;A_{rs}\iidsim \chi_1^2.$

Since by Lemma~\ref{lem1.sec2}, we have $T_n=\tT_n+\mathcal{O}(n^{-1/2})$. This means that  $T_n$ and $\tT_n$ have the same asymptotic distribution. That is, $T_n\convL \tT$ with $\tT$ defined above.
\end{proof}

\begin{proof}[Proof of Theorem~\ref{muvar2.thm}] By Lemma~\ref{lem1.sec2}, we have $T_n=\tT_n+\mathcal{O}(n^{-1/2})$. It follows that $\E(T_n)=\E(\tT_n)+\mathcal{O}(n^{-1/2})$, $\Var(T_n)=\Var(\tT_n)+\mathcal{O}(n^{-1/2})$, and $\E[T_n-\E(T_n)]^3=\E[\tT_n-\E(\tT_n)]^3+\mathcal{O}(n^{-1/2})$. Thus, we just need to find $\E(\tT_n), \Var(\tT_n)$, and $\E[\tT_n-\E(\tT_n)]^3$. Note that
\[
\tT_n=n^{-1}\limsum_{i=1}^n[\tK(x_i,x_i)\tL(y_i,y_i)]+2n^{-1}\limsum_{1\le i<j\le n}[\tK(x_i,x_j)\tL(y_i,y_j)].
\]
Since  $\E_x[\tK(x,x')]=\E_{x'}[\tK(x,x')]=\E_{x,x'}[\tK(x,x')]=0$ and $\E_y[\tL(y,y')]=\E_{y'}[\tL(y,y')]=\E_{y,y'}[\tK(y,y')]=0$, we have
\[
\E(\tT_n)=n^{-1}\limsum_{i=1}^n\E[\tK(x_i,x_i)\tL(y_i,y_i)]=\E[\tK(x,x)]\E[\tL(y,y)]=M_1N_1.
\]
It follows that
\[
\begin{array}{rcl}
\tT_n-\E(\tT_n)&=&n^{-1}\limsum_{i=1}^n\{\tK(x_i,x_i)\tL(y_i,y_i)-M_1N_1\}+2n^{-1}\limsum_{1\le i<j\le n}[\tK(x_i,x_j)\tL(y_i,y_j)]\\
               &=&n^{-1}\limsum_{i=1}^n A_i+2n^{-1}\limsum_{1\le i<j\le n}B_{ij}.
\end{array}
\]
Note that $\E(A_i)=0$ and $A_1,A_2,\cdots,A_n$ are i.i.d. and under the null hypothesis, we have
\[
\E_{x_iy_i}(B_{ij})=\E_{x_jy_j}(B_{ij})=\E(B_{ij})=0.
\]
Thus
\[
\begin{array}{rcl}
&&\Var(\tT_n)=n^{-1}\Var(A_1)+4n^{-2}\limsum_{1\le i<j\le n}\E(B_{ij}^2)\\
           &=&n^{-1}\Var[\tK(x,x')]\Var[\tL(y,y')]+2(1-n^{-1})\E[\tK^2(x,x')]\E[\tL^2(y,y')]\\
           &=&2M_2N_2+\mathcal{O}(n^{-1}).
\end{array}
\]
%
Finally,
\[
\begin{array}{rcl}
 \E[\tT_n-\E(\tT_n)]^3&=&n^{-3}\E\left(\limsum_{i=1}^n A_i\right)^3+8n^{-3}\E\left(\limsum_{1\le i<j\le n}B_{ij}\right)^3\\
           &=&n^{-2}\E(A_1^3)+8n^{-3}\E\left(\limsum_{1\le i<j\le n}B_{ij}\right)^3,
\end{array}
\]
where $\E(A_1^3)=\E\left[\tK(x,x)\tL(y,y)-M_1N_1\right]^3$ and
\[
\begin{array}{rcl}
   \E\left(\limsum_{1\le i<j\le n}B_{ij}\right)^3&=&\E\left[\sum\limits_{i<j}B_{ij}^3+3\sum^{*}B_{ij}^2B_{\alpha\beta}+6\sum^{**} B_{ij} B_{\alpha\beta}B_{uv}\right]\\
               &=&\left\{\frac{n(n-1)}{2}\E[B_{12}^3]
               +6\frac{n(n-1)(n-2)}{3!}\E[B_{12}B_{23}B_{31}]\right\}\\
               &=&n(n-1)(n-2)M_3N_3+\frac{n(n-1)}{2}\E[\tK^3(x,x')]\E[\tL^3(y,y')],
\end{array}
\]
where  $*$ means ``$i<j, \alpha<\beta$'' and ``$(i,j)\neq (\alpha,\beta)$'' while $**$ means ``$i<j, \alpha<\beta, u<v$'' and ``$(i,j), (\alpha, \beta), (r,s)$ are not  mutually equal to each other.''
It follows that  $\E[\tT_n-\E(\tT_n)]^3=8M_3N_3+\mathcal{O}(n^{-1})$. The theorem is then proved.
%
\end{proof}

\begin{proof}[Proof  of (\ref{M23.sec2})] Notice that
\[
\left\{\limsum_{i=1}^n\limsum_{j=1}^n [\tK^*(x_i,x_j)]^2-\limsum_{i=1}^n[\tK^*(x_i,x_i)]^2\right\}=\tr(\tbK^{*2})-\tr(\tbK^*o\tbK^*).
\]
It follows that $\hM_2=[n(n-1)]^{-1}[\tr(\tbK^{*2})-\tr(\tbK^*o\tbK^*)]$. Further, we have
\[
\begin{array}{rcl}
&&\qquad \qquad \limsum\limits_{1\le i<j<k\le n}[\tK^*(x_i,x_j)\tK^*(x_j,x_k)\tK^*(x_k,x_i)]\\
&=&\limsum_{i=1}^n\limsum_{j=1}^n\limsum_{\ell=1}^n\tK^*(x_i,x_j)\tK^*(x_j,x_k)\tK^*(x_k,x_i)-\limsum_{i=1}^n[\tK^*(x_i,x_i)]^3-3\limsum_{i\neq j}\tK^*(x_i,x_j)\tK^*(x_j,x_i)\tK^*(x_i,x_i)\\
&=&\tr(\tbK^{*3})-\tr(\tbK^*o\tbK^*o\tbK^*)-3\left\{\limsum_{i=1}^n\limsum_{j=1}^n\tK^*(x_i,x_j)\tK^*(x_j,x_i)\tK^*(x_i,x_i)-\limsum_{i=1}^n[\tK^*(x_i,x_j)]^3\right\}\\
&=&\tr(\tbK^{*3})-\tr(\tbK^*o\tbK^*o\tbK^*)-3\left\{\tr[\diag(\tbK^*)\tbK^{*2}]-\tr(\tbK^*o\tbK^*o\tbK^*)\right\}\\
&=&\tr(\tbK^{*3})-3\tr[\diag(\tbK^*)\tbK^{*2}]+2\tr(\tbK^*o\tbK^*o\tbK^*).
\end{array}
\]
It follows that $\hM_3=[n(n-1)(n-2)]^{-1}\left\{\tr(\tbK^{*3})-3\tr[\diag(\tbK^*)\tbK^{*2}]+2\tr(\tbK^*o\tbK^*o\tbK^*)\right\}$.
\end{proof}

\begin{proof}[Proof of Theorem~\ref{hcoefs.thm}] Under the condition (\ref{cond.sec2}),  by Lemma~\ref{lem1.sec2},
	we have $\tilde{K}^*(x_i,x_j)=\tilde{K}(x_i,x_j)+\mathcal{O}(n^{-1/2})$ uniformly for all $x_i,x_j$'s. Since $|\tilde{K}(x,x')|\le 4B_K<\infty$ for all $x,x'\in\mathcal{X}$, by (\ref{hMNexp1.sec2}),  we have
\[
\hM_1=\tilde{M}_1+\mathcal{O}(n^{-1/2}),\;\;
\hM_2=\tilde{M}_2+\mathcal{O}(n^{-1/2}),\;\;\mbox{ and }\;\;
\hM_3=\tilde{M}_3+\mathcal{O}(n^{-1/2}),
\]
where
\[
\begin{array}{c}
\tilde{M}_1=\frac{1}{n}\sum\limits_{i=1}^n \tK(x_i,x_i),\;\;
\tilde{M}_2=\frac{2}                                                                                                                                                                           {n(n-1)}\sum_{1\le i<j\le n} \tilde{K}^2(x_i,x_j),\\
\tilde{M}_3=\frac{6}{n(n-1)(n-2)} \sum_{1\le i<j<k\le n}\tilde{K}(x_i,x_j)\tilde{K}(x_j,x_k)\tilde{K}(x_k,x_i).
\end{array}
\]
Since $\tilde{M}_1, \tilde{M}_2, \tilde{M}_3$ are U-statistics for $M_1, M_2, M_3$ respectively and under the condition (\ref{cond.sec2}),  we have
\[
\begin{array}{c}
\E[\tK(x,x)]^2\le  (4B_K)^2<\infty,\;\;
\E\left[\tilde{K}(x,x')\right]^4\le (4B_K)^4<\infty,\\
\E\left[\tilde{K}(x,x')\tilde{K}(x', x'')\tilde{K}(x'',x)\right]^2\le (4B_K)^6<\infty.
\end{array}
\]
Then by Lemma A of \cite[p.185]{serfling}, as $n\to\infty$,  we have $\tilde{M}_1\stackrel{p}{\longrightarrow} M_1, \tilde{M}_2\stackrel{p}{\longrightarrow} M_2$, and   $\tilde{M}_3\stackrel{p}{\longrightarrow} M_3$.
It follows that as $n\to\infty$, we have $\hM_1\stackrel{p}{\longrightarrow} M_1, \tilde{M}_2\stackrel{p}{\longrightarrow} M_2$, and $\tilde{M}_3\stackrel{p}{\longrightarrow} M_3$. Thus,
as $n\to\infty$, we have $\hat{M}_{\ell} \stackrel{p}{\longrightarrow} M_{\ell}, \ell=1, 2,3$. Similarly, we can show that as $n\to\infty$, we have $\hat{N}_{\ell} \stackrel        {p}{\longrightarrow} N_{\ell}, \ell=1, 2,3$. The remaining claims then follow. The theorem is complete.
\end{proof}

\begin{proof}[Proof of Theorem~\ref{power.thm}] Under the condition (\ref{cond.sec2}), by Lemma~\ref{lem1.sec2}, we have $T_n=\tT_n+\mathcal{O}(n^{-1/2})$ and hence $T_n/n=\tT_n/n+\mathcal{O}(n^{-3/2})$. It follows that $\E(T_n/n)=\E(\tT_n/n)+\mathcal{O}(n^{-3/2})=\HSIC+\mathcal{O}(n^{-3/2})$. Thus, $\sqrt{n}\left(T_n/n-\HSIC\right)=\sqrt{n}\left(T_n/n-\HSIC\right)+\mathcal{O}(n^{-1})$. Now
\[
\begin{array}{rcl}
\tT_n/n&=&n^{-2}\limsum_{i=1}^n[\tK(x_i,x_i)\tL(y_i,y_i)]+\frac{n-1}{n}\binom{n}{2}^{-1}\limsum_{1\le i<j\le n}[\tK(x_i,x_j)\tL(y_i,y_j)]\\
&=&\binom{n}{2}^{-1}\limsum_{1\le i<j\le n}[\tK(x_i,x_j)\tL(y_i,y_j)]+\mathcal{O}(n^{-1}).
\end{array}
\]
Let $(x',y')$ be an independent copy of $(x,y)$. Since $\E[\tK(x,x')\tL(y,y')]=\HSIC>0$, $x$ and $y$ are not independent. It follows that $\E[\tK(x,x')\tL(y,y')|(x',y')]\neq 0$ and
$\sigma^2=\Var\left\{\E[\tK(x,x')\tL(y,y')|(x,y)]\right\}>0$, by \citet[Theorem A, Sec. 5.5.1]{serfling}, as $n\to\infty$, we have
$
\sqrt{n}\left(\tT_n/n-\HSIC\right)\convL N(0,4\sigma^2).
$
It follows that $\sqrt{n}\left(T_n/n-\HSIC\right)\convL N(0,4\sigma^2)$. Under the condition (\ref{cond.sec2}), by Theorem~\ref{hcoefs.thm}, we have $\hbeta_0\convp \beta_0, \hbeta_1\convp \beta_1$ and $\hd\convp d$. Let $\Phi(\cdot)$ denote the cumulative distribution function of the standard normal distribution.
Therefore, under the condition (\ref{cond.sec2}) and the local alternative (\ref{localt.sec2}), when $n$ is large, for  any significance level $\alpha$,  we have
\[
\begin{array}{rcl}
Pr\left[T_n\ge \hbeta_0+\hbeta_1\chi_{\hd}^2(\alpha)\right]&=&Pr\left[ \frac{\sqrt{n}\left(T_n/n-\HSIC\right)}{2\sigma}
\ge \frac{\hbeta_0+\hbeta_1\chi_{\hd}^2(\alpha)}{2\sqrt{n}\sigma}-\frac{\sqrt{n}\HSIC}{2\sigma}\right]\\
                     &=&\Phi\left[-\frac{\beta_0+\beta_1\chi_{d}^2(\alpha)}{2\sqrt{n}\sigma}+\frac{n^{\Delta} h}{2\sigma} \right][1+o(1)]\\
                     &=&\Phi\left[n^{\Delta} h/(2\sigma) \right][1+o(1)],
\end{array}
\]
which tends to $1$ as $n\to\infty$. The theorem is proved.
\end{proof}

\bibliographystyle{apalike}
\bibliography{two_sam}

\end{document}